\documentclass[12pt,a4paper,amsmath,amssymb]{amsart}

\usepackage{amsthm}
\usepackage{latexsym}
\usepackage{amsfonts}
\usepackage{bbm,bm,dsfont}
\usepackage{color} 

\newtheorem{theorem}{Theorem}

\newtheorem{remark}{Remark}
\newtheorem{corollary}{Corollary}
\newtheorem{proposition}{Proposition}
\newtheorem{example}{Example}

\newcommand{\blue}[1]{\textcolor{blue}{#1}}

\newcommand{\N}{\mathbb N}
\newcommand{\R}{\mathbb R}
\newcommand{\Z}{\mathbb Z}
\newcommand{\C}{\mathbb C}

\newcommand{\T}{\mathbb T}

\newcommand{\hi}{\mathcal{H}} %Hilbert space
\newcommand{\hd}{{\mathcal{H}_\oplus}} %direct integral Hilbert space?
\newcommand{\ki}{\mathcal{K}} %Hilbert space
\newcommand{\li}{\mathcal{L}} %other Hilbert space

\newcommand{\lh}{\mathcal{L(H)}} %bounded linear operators
\newcommand{\lk}{\mathcal{L(K)}} %bounded linear operators
\renewcommand{\th}{\mathcal{T(H)}} %trace class operators
\newcommand{\tk}{\mathcal{T(K)}} %trace class operators
\newcommand{\sh}{\mathcal{S(H)}} %states
\newcommand{\tr}[1]{\mathrm{tr}\left[#1\right]} %trace
\newcommand{\kb}[2]{|#1\,\rangle\langle\,#2|} %ketbra
\newcommand{\Mo}{\mathsf{M}} %generic observable
\newcommand{\sfe}{\mathsf{E}} %generic observable
\newcommand{\Po}{\mathsf{P}} %generic observable
\newcommand{\M}{\mathcal M} %projection
\renewcommand{\O}{\mathrm{Obs}} %observables
\newcommand{\In}{\mathrm{Ins}} %observables
\newcommand{\ext}{\mathrm{Ext}\,} %extremals
\newcommand{\A}{\mathsf{A}}

\newcommand{\K}{{\mathbf K}(X,V)}

\def\<{\langle}
\def\>{\rangle}
\def\d{{\mathrm d}}
\newcommand{\bo}[1]{\mathcal{B}(#1)} %Borel sigma-algebra
\newcommand{\fii}{\varphi}
\newcommand{\la}{\lambda}
\newcommand{\CHI}[1]{\ensuremath{ \chi\raisebox{-1ex}{$\scriptstyle #1$} }}
\newcommand{\CHII}[1]{\ensuremath{{\hat\chi}\raisebox{-1ex}{$\scriptstyle #1$} }}
\newcommand{\ov}{\overline}

\newcommand{\lin}{{\rm lin}}
\newcommand{\e}{{\bf h}}
\renewcommand{\k}{{\bf k}}

\begin{document} 

\title{Completely positive maps on modules, instruments, extremality problems, and applications to physics}

\author{Juha-Pekka Pellonp\"a\"a}
\email{juhpello@utu.fi}
\address{Turku Centre for Quantum Physics, Department of Physics and Astronomy, University of Turku, FI-20014 Turku, Finland}

\begin{abstract}
Convex sets of completely positive maps and positive semidefinite kernels are considered in the most general context of modules over $C^*$-algebras and a complete charaterization of their extreme points is obtained. As a byproduct, we determine extreme quantum instruments, preparations, channels, and extreme autocorrelation functions. Various applications to quantum information and measurement theories are given. The structure of quantum instruments is analyzed thoroughly.

%\newline

%\noindent
%PACS numbers: 03.65.Ta, 03.67.--a
\end{abstract}

\maketitle

%%%%%%%%%%%%%%%%%%%%%%%%%%%%%%%%%%%%%%%%%%%%%%%%%%%%%%%%%%%

\section{Introduction}

There is no question about the importance of completely positive maps (instruments, POVMs, channels) in quantum information and measurement theories \cite{BuLaMi,Da,Teikonkirja,Ho,Kr}.
One could say that they form a core of modern mathematical analysis of quantum theory.
For example, a normalized positive operator valued measure (POVM) describes the statistics of the outcomes of a quantum measurement and can be identified with a quantum observable.

In 1970, Davies and Lewis \cite{DaLe} introduced the concept of instrument which turned out to be crucial in developing quantum measurement theory since,
besides measurement statistics,
it also describes the state change due to a quantum measuring process. In 1984,
Ozawa \cite{Oz84} proved that any completely positive (CP) instrument can be dilated to a quantum measuring process, that is, any instrument can be realized as a measurement model of a POVM. The role of completely positivity is essential in this characterization, see also \cite{BuLa}.
Later Holevo \cite{Ho98} analyzed the structure of instruments and showed that any CP instrument has a pointwise Kraus decomposition \cite{Kr}.
Recently, applications of quantum instruments and their extremality problems have been studied extensively, see e.g.\ \cite{CaHeTo,ChDAPe,DAPeSe,HaLaSc,HeWo,Ho98,Yl} and references therein.

Since instruments (and hence POVMs and channels) are special cases of CP maps, their mathematical analysis is based on the celebrated Stinespring dilation theorem \cite{St}. This theorem has been generalized in many directions, the most general extension being \cite[Theorem 4.3]{PeYl} where the CP maps are defined on a (unital) $C^*$-algebra and get their values in the vector space of $A$-sesquilinear $A$-valued forms on an $A$-module $V$ (where $A$ is a $C^*$-algebra). In quantum mechanics, one typically chooses $A=\C$ and let $V$ be a Hilbert space, but in geometric theories of physics (general relativity, gauge field theory, etc.)
one uses more complicated algebras $A$.
For example, when theory is based on a vector bundle structure over a  manifold $\Omega$,
one may take $A=C_0(\Omega)$, the continuous functions $\Omega\to\C$ vanishing at infinity. Then $V$ could be the linear space of continuous vector fields.
This type of structures (especially Hilbert $C^*$-modules \cite{Manuilov}) 
are used in noncommutative geometry \cite{Gracia} which forms a link between geometric theories and quantum theory.  
Sometimes noncommutative geometry is viewed as a route to quantum gravity and spacetime. A nice application of CP maps to the problem of quantum spacetime is given in \cite{HaMaSa}.

In this article, we define convex sets of CP maps and positive definite kernels in the most general context and charaterize completely their extreme points. We apply this result e.g.\ to arbitrary CP instruments. 
Here are the results of this paper (some of them are known in the discrete finite-dimensional cases but our results are also valid in `nondiscrete' cases and in infinite dimensions):

\begin{itemize}
\item The structure of an arbitrary instrument $\M$ is determined in several different ways (Theorem \ref{th2}) by using e.g.\ structure vectors $\psi_m^t(x)$, generalized vectors $d_k^t(x)$, pointwise Kraus operators $\A_k(x)$, and setwise Kraus operators $\A_k(X)$. See also the Appendix.

\item Any instrument has a {\it minimal} pointwise Kraus decomposition (item (2) of Theorem \ref{th2}).

\item The Dirac formalism is extended to instruments (Remark \ref{dirac}) so that it can be used to find compatible instruments of POVMs (Section \ref{luku4}).

\item We characterize the extreme points of the convex set of instruments (see (4) and (5) of Theorem \ref{th2} and Remark \ref{exrem}).

\item The extreme point characterization of instruments is applied to observables, preparations, the discrete case, channels, and the finite dimensional case (Subsection \ref{examples}).

\item Extreme instruments are discrete in finite dimensions (Proposition \ref{lksjdf}).

\item For any POVM $\Mo$, we show that the $\Mo$-compatible instruments can be identified with the decomposable CP channels (Theorem \ref{compa}) and can be viewed as  combinations of L\"uders operations and channels (Corollary \ref{corolla1}).

\item We present a complete characterization for pure realizations (measurement models) of instruments (Theorem \ref{purereali} and Remark \ref{rem7}), and for minimal pure realizations (Corollary \ref{purecorolla}).

\item The standard model of quantum measurement theory is generalized for arbitrary POVMs (Example \ref{stmo}).

\item We determine the posterior (i.e.\ the post measurement) states (of a measurement) for arbitrary input states by using a minimal pointwise Kraus form of an instrument (Subsection \ref{postpost}).

\item Any instrument can be maximally refined into a rank-1 instrument, and if an instrument is extreme then its  rank-1 refinement is also extreme (Proposition \ref{jgjhfjoooot}).

\item For any rank-1 POVM $\Mo$, we prove that the $\Mo$-compatible instruments are all nuclear and their associate channels are entanglement-breaking (Theorem \ref{hnasgcbnsjcd} and Example \ref{nuclearex}).

\item We determine the extreme points of the convex set of `very general' positive definite kernels (Theorem \ref{extker}) and, as an application, characterize extreme autocorrelation functions of stochastic processes (Proposition \ref{auto}).

\item A complete characterization of the extreme points of the convex set of `very general' CP maps on modules is given in Theorem \ref{seiska}.

\item Finally, we present a generalization for Choi isomorphism widely used in quantum information (Theorem \ref{CJyleistys}).

\end{itemize}

\section{Basic notations and definitions}

For any Hilbert space $\hi$ we let $\lh$ [resp.\ $\th$] denote the set of bounded [resp.\ trace-class] operators on $\hi$. 
We let the innerproduct $\<\,\cdot\,|\,\cdot\,\>$ of a Hilbert space (or any sesquilinear form) be linear with respect to its second argument.
We say that a positive operator $\rho\in\th$ of trace 1 is a state (or a density operator) and denote the set of states by $\sh$.
The identity operator of any Hilbert space $\hi$ is denoted by $I_\hi$.
Throughout this article, we let $\hi$ and $\ki$ be {\it separable} (complex) nontrivial Hilbert spaces
and $(\Omega,\Sigma)$ be a measurable space (i.e.\ $\Sigma$ is a $\sigma$-algebra of subsets of a set $\Omega$). If a fixed measure $\mu$ is given on $(\Omega,\Sigma)$, without restricting generality, we assume that $\Sigma$ is complete with respect to $\mu$ (i.e.\ contains $\mu$--null sets). Hence, the concepts of $\Sigma$--measurability and $\mu$--measurability coincide and we may just speak about measurability of a function $f:\,\Omega\to\C$.
As usual we define an empty sum to be 0, e.g.\ $\sum_{k=1}^0(\ldots):=0$, and $\N:=\{0,1,\ldots\}$. Moreover, %for each $k\in\N$, we let $\Z_k$ denote the set $\{1,2,\ldots,k\}$, and 
$\N_\infty:=\N\cup\{\infty\}$ and $\N_+:=\{1,2,\ldots\}$.

\subsection*{Operator measures}
Let $\Mo:\,\Sigma\to\lh $ be an {\it operator (valued) measure,} i.e.\ (ultra)weakly $\sigma$-additive mapping.
We call $\Mo$ \emph{positive} if for all $X\in\Sigma$,  $\Mo(X)\geq 0$, \emph{normalized} if $\Mo(\Omega)=I_\hi$, and \emph{projection valued} if $\Mo(X)^2=\Mo(X)^*=\Mo(X)$ for all $X\in\Sigma$.
Normalized positive operator valued measures (POVMs) are identified with \emph{(quantum) observables} whereas normalized projection valued measures (PVMs) are called \emph{spectral measures} or \emph{sharp observables}. The convex set of POVMs $\Mo:\Sigma\to\lh $ is denoted by $\O(\Sigma,\,\hi)$ and its extreme points by $\ext\O(\Sigma,\,\hi)$.
A convex combination (observable) $t\Mo_1+(1-t)\Mo_2$, $0<t<1$, can be viewed as a randomization of measuring procedures represented by the observables $\Mo_1$ and $\Mo_2$. An extreme observable $\Mo\in\ext\O(\Sigma,\,\hi)$ cannot be obtained as a (nontrivial) combination; this means that the measurement of $\Mo$ involves no redundancy caused by mixing different measuring schemes.

\subsection*{Instruments}
We say that a map $\M:\,\Sigma\times \lk\to \lh$ is a {\it (CP quantum) instrument} if
\begin{enumerate}
\item for all $X\in\Sigma$, the mapping $\lk\ni B\mapsto \M(X,B)\in\lh$ is linear, {\it completely positive} (CP), and ultraweakly continuous (normal), 
\item $\M(\Omega,I_\ki)=I_\hi$, %[and $0\le \M(X,I_\ki)\le I_\hi$ for all $X\in\Sigma$],
\item $\tr{\rho\M(\cup_{i=1}^\infty X_i,B)}=\sum_{i=1}^\infty\tr{\rho\M(X_i,B)}$ for any disjoint sequence $\{X_i\}_{i=1}^\infty\subseteq\Sigma$ and for all $\rho\in\th$, $B\in\lk$.
\end{enumerate}
%(Usually, in physical applications, $\ki=\hi$.)
For any $B\in\lk$,
we define  an operator measure $$\Mo_B:\,\Sigma\to\lh,\,X\mapsto \Mo_B(X):=\M(X,B).$$ 
It is positive if $B\ge 0$ and normalized if $B=I_\ki$. Hence, $\Mo_{I_\ki}$ is a POVM, the {\it associate observable of $\M$.}
An instrument can be seen as a certain collection of operator measures (indexed by bounded operators $B$).
Any $B\in\lk$ can be (nonuniquely) decomposed into positive parts, that is, $B=\sum_{k=0}^3 i^kB_k$
where operators $B_k$ are bounded and positive. 
%Obviously, such a decomposition is not unique so that we must make a choice for the above decomposition as follows: First, write $B=B_{\rm Re}+i B_{\rm Im}$ where the operators $B_{\rm Re}:=\frac12(B+B^*)$ and $B_{\rm Im}:=\frac1{2i}(B-B^*)$ are self-adjoint. Secondly, decompose any self-adjoint operator $C=C^*$ into natural positive parts, that is, write $C=C_+-C_-$ where $C_\pm$ are positive and defined as $C_+:=\int_0^\infty x\,\d\mathsf C(x)$ and $C_-:=-\int_{-\infty}^0 x\,\d\mathsf C(x)$ where $\mathsf C$ is the spectral measure of $C$.
%For example, if $\fii,\,\psi\in\ki$,
%$$
%\kb\fii\psi= 
%\frac14\sum_{k=0}^3 i^k\kb{\fii+i^k\psi}{\fii+i^k\psi}.
%$$ 
Hence, by linearity, any instrument $\M$ is uniquely determined already by positive operator measures $\Mo_B$ where $B\ge0$.

The convex set of instruments $\M:\,\Sigma\times \lk\to \lh$ is denoted by $\In(\Sigma,\,\ki,\,\hi)$ and its extreme points by $\ext\In(\Sigma,\,\ki,\,\hi)$. Recall that any $\M\in\In(\Sigma,\,\ki,\,\hi)$ defines a `predual map' $\M_*:\,\Sigma\times\th\to\tk$ by $\tr{\M_*(X,\rho)B}:=\tr{\rho\M(X,B)}$ for all
$X\in\Sigma$, $\rho\in\th$, and $B\in\lk$. 
Sometimes $\M$ is referred as a Heisenberg instrument and $\M_*$ a Schr\"odinger instrument.
Obviously, the map $\M\mapsto\M_*$ is an affine bijection so that $\M$ is an extreme Heisenberg instrument if and only if $\M_*$ is an extreme Schr\"odinger instrument.

%\begin{enumerate}
%\item for all $X\in\Sigma$, the mapping $\th\ni \rho\mapsto \M_*(X,\rho)\in\tk$ is linear, CP, and (trace norm) continuous, 
%\item $\tr{\M_*(\Omega,\rho)}=\tr\rho$ for all $\rho\in\th$,
%\end{enumerate}

\begin{remark}\rm
On the first hand,
any instrument $\M$ defines a CP channel,\footnote{A map $T:\,\lk\to\lh$ is a {\it (quantum) operation} if it is linear, positive, ultraweakly continuous, and $T(I_{\ki})\le I_\hi$. If, moreover, $T(I_\ki)=I_\hi$ then $T$ is said to be a {\it (quantum) channel.} 
The maps $B\mapsto\M(X,B)$ are CP operations.}
 $B\mapsto\M(\Omega,B)$,  the {\it associate channel of $\M$}.
On the other hand, if $T:\,\lk\to\lh$ is a CP channel, then by choosing
$\Omega=\{0\}$ and $\Sigma=2^{\{0\}}=\big\{\emptyset,\{0\}\big\}$, one can define an instrument 
$\M_T(\{0\},B):=T(B)$, $B\in\lk$.
Similarly, for any POVM $\Mo:\,\Sigma\to\lh$
there exist an instrument $\M^\Mo:\,\Sigma\times\mathcal \C\to\lh$ defined by
$\M^\Mo(X,c):=c\, \Mo(X)$ where $c\in\C\cong\mathcal L(\C)$ (via $c\mapsto c\kb11$).
We call the instruments $\M_T$ and $\M^\Mo$ {\it trivial instruments} associated with $T$ and $\Mo$, respectively.
Thus, it follows that all general results for instruments are applicable to channels and POVMs.

Finally, recall the following nontrivial result of Davies and Lewis \cite[Theorem 1]{DaLe}: For any POVM $\Mo\in\O(\Sigma,\,\hi)$ there exists a (nonunique) CP instrument $\M\in\In(\Sigma,\,\hi,\,\hi)$, an {\it $\Mo$-compatible instrument}, such that its associate observable $\Mo_{I_\hi}=\Mo$, that is, 
$\tr{\M_*(X,\rho)}\equiv\tr{\rho\Mo(X)}$. 
Let $\M\in\In(\Sigma,\,\hi,\,\hi)$ be an $\Mo$-compatible instrument of a projection valued measure $\Mo\in\O(\Sigma,\,\hi)$.
Then $\M(X,B)\equiv\Mo(X)T(B)$ where $T$ is a CP channel such that $\Mo(X)T(B)\equiv T(B)\Mo(X)$ \cite[Prop.\ 4.3 and 4.4]{Oz84}.
It is then obvious  that for a POVM $\Mo$ and a channel $T$ there does not necessarily exist an instrument $\M$ such that $\M(X,I_\ki)\equiv\Mo(X)$ and $\M(\Omega,B)\equiv T(B)$. %(see also example \ref{POVM&channel}).
\end{remark}

\section{Diagonalization and extremality results for POVMs and instruments}

Let $\e=\{h_n\}_{n=1}^{\dim\hi}$ %and $\e'=\{e'_m\}_{m=1}^{\dim\hi}$ 
be an orthonormal (ON) basis of $\hi$ 
%We define an equivalence relation $\sim$ as follows: $\e\sim\e'$ if and only if any $e'_m$ is a {\it finite} linear combination of vectors $h_n$. Let $[\e]$ be the corresponding equivalence class of bases.
%If $\e'\in[\e]$ we say that $\e'$ is an $\e$--admissible basis. 
and $$V_\e:=\lin_\C\{h_n\,|\,1\le n<\dim\hi+1\}.$$ %so that $V_\e=V_{\e'}$ if $\e\sim\e'$ and we have a well-defined mapping $[\e]\mapsto V_\e$. 
Note that $V_\e$ is dense in $\hi$.
Let $V_\e^\times$ be the algebraic antidual of the vector space $V_\e$. Recall that $V_\e^\times$ can be identified with the linear space of formal series $c=\sum_{n=1}^{\dim\hi}c_nh_n$ where $c_n$'s are arbitrary complex numbers. 
Hence, $V_\e\subseteq\hi\subseteq V_\e^\times$.
Denote the dual pairing $\<\psi|c\>:=\sum_{n=1}^{\dim\hi} \<\psi|h_n\>c_n$ and $\<c|\psi\>:=\overline{\<\psi|c\>}$ for all $\psi\in V_\e$ and $c\in V_\e^\times$.
We say that a mapping $c:\,\Omega\to V_\e^\times,\,x\mapsto \sum_{n=1}^{\dim\hi}c_n(x)h_n$ is {\it (weak${}^*$-)measurable} if its components $x\mapsto c_n(x)$ are measurable. Note that, if $c:\,\Omega\to\hi\subseteq V_\e^\times$ is weak${}^*$-measurable then
the maps $x\mapsto\<\psi|c(x)\>$ are measurable for all $\psi\in\hi$.
For any linear map $A:\,V_\e\to H$, where $H$ is a Hilbert space, we let $A^*$ denote the adjoint (transpose) linear map from $H$ to $V_\e^\times$ defined by $\<\psi|A^*\fii\>:=\<A\psi|\fii\>$, $\psi\in V_\e$, $\fii\in H$.
Note that $A^*$ is not necessarily the usual Hilbert space adjoint of $A$ given by the Fr\'echet-Riesz representation theorem. However, for bounded operators between Hilbert spaces and for elements of $C^*$-algebras we use the same symbol ${}^*$ for the usual adjoint and involution.

Let $\hd$ denote a direct integral $\int_\Omega^\oplus\hi_{n(x)}\d\mu(x)$
of {separable} Hilbert spaces $\hi_{n(x)}$ such that $\dim\hi_{n(x)}=n(x)\in\N_\infty$; here $\mu$ is a $\sigma$-finite nonnegative measure\footnote{Note that $\mu$ can be a probability measure everywhere in this paper; any $\sigma$-finite measure is equivalent with a probability measure.} 
on $(\Omega,\Sigma)$. Let $\li$ be an infinite dimensional separable Hilbert space with an ON basis ${\bf b}=\{b_n\}_{n=1}^\infty$.
By choosing a measurable field of ON bases, $\{b_n(x)\}_{n=1}^{n(x)},$ one gets a decomposable unitary operator $U:\,\hd\to\hi_\oplus'$, $U(x)b_n(x):=b_n$, where $\hi_\oplus':=\int_\Omega^\oplus\hi'_{n(x)}\d\mu(x)$ with fibers $$\hi'_{n(x)}:={\lin_\C\{b_n\,|\,1\le n\le n(x)\}}$$
if $0<n(x)<\infty$, $\hi'_0:=\{0\}$, and $\hi'_\infty:=\li$.
Hence, without restricting generality, we simply assume that $\hi_{n(x)}\equiv\hi'_{n(x)}$ and thus $\hd=\hi_\oplus'$. Now $\hd$ can be considered as a closed subspace of $L^2(\mu,\li)\cong L^2(\mu)\otimes\li$, the $\mu$-square integrable functions $\Omega\to\li$, 
and one has a decomposable projection $P=\int^\oplus_\Omega P(x)\d\mu(x)$ from $L^2(\mu,\li)$ onto $\hd$, where $P(x)=\sum_{n=1}^{n(x)}\kb{b_n}{b_n}$. We say that $\hd$ {\it is embedded in $L^2(\mu,\li)$} and write 
$\hd\subseteq L^2(\mu,\li)$.

For each $f\in L^\infty(\mu)$, we denote briefly by $\hat f$ the multiplicative (i.e.\ diagonalizable) bounded operator $(\hat f\psi)(x):=f(x)\psi(x)$ on {\it any} direct integral Hilbert space $\hd=\int_\Omega^\oplus\hi_{n(x)}\d\mu(x)$. Especially, one has the {\it canonical spectral measure} $\Sigma\ni X\mapsto\CHII X\in\mathcal L(\hd)$ (where $\CHI X$ is the characteristic function of $X\in\Sigma$).
We will use the following proposition \cite[Theorem 1, p.\ 187]{Di} several times:
\begin{proposition}\label{prop1}
Let $\hi^i_\oplus=\int_\Omega^\oplus\hi^i_{n^i(x)}\d\mu(x)$, $i=1,\,2$, be two direct integral Hilbert spaces, and let $D:\,\hi^1_\oplus\to\hi^2_\oplus$ be a bounded operator. Then, $D\CHII X=\CHII X D$ for all $X\in\Sigma$ if and only if $D\hat f=\hat f D$ for all $f\in L^\infty(\mu)$ if and only if $D$ is decomposable, i.e.\ $D=\int_\Omega^\oplus D(x)\d\mu(x)$ where $D(x):\,\hi^1_{n^1(x)}\to\hi^2_{n^2(x)}$ are bounded and, for any $\psi\in\hi^1_\oplus$, $(D\psi)(x)=D(x)\psi(x)$ for $\mu$-almost all $x\in\Omega$, and
$\|D\|=\mu\text{\rm-ess sup}_{x\in\Omega}\|D(x)\|<\infty$.
\end{proposition}

\subsection{Observables}\label{POVM}

We have the following theorem proved in \cite{HyPeYl,Pe11}:

Let $\Mo:\,\Sigma\to\lh$ be a positive operator measure and $\mu:\,\Sigma\to[0,\infty]$ a $\sigma$-finite measure such that $\Mo$ is absolutely continuous with respect to $\mu$. Let $\e$ be an ON basis of $\hi$.
\begin{theorem} \label{th1}
The exists  a direct integral $\hd=\int_\Omega^\oplus\hi_{n(x)}\d\mu(x)$ (with $n(x)\le\dim\hi$) 
such that,
for all $X,\,X'\in\Sigma$,
\begin{enumerate}
\item
$\Mo(X)=Y^*\CHII X Y$ where $Y=\sum_{m=1}^{\dim\hi}\kb{\psi_m}{h_m}$ is a bounded operator and
$\{\psi_m\}_{m=1}^{\dim\hi}\subseteq\hd$ is such that
the set of linear combinations of vectors $\CHI X \psi_m$ is dense in $\hd$
(a minimal Naimark dilation for $\Mo$). Hence, by defining $Y(X):=\CHII X Y=\sum_{m=1}^{\dim\hi}\kb{\CHI X\psi_m}{h_m}$,
\begin{eqnarray*}
\Mo(X\cap X')=Y(X)^*Y(X')=
\sum_{n,m=1}^{\dim\hi}\int_{X\cap X'}\<\psi_n(x)|\psi_m(x)\>\d\mu(x)\kb{h_n}{h_m} 
\end{eqnarray*}
weakly (a minimal Kolmogorov decomposition for $\Mo$). 
\item
There are measurable maps $d_k:\,\Omega\to V_\e^\times$ such that, for all $x\in\Omega$, 
the vectors $d_k(x)\ne 0$, $k<n(x)+1$ are linearly independent, and
$$
\<\fii|\Mo(X)\psi\>=\int_X \sum_{k=1}^{n(x)} \<\fii|d_k(x)\>\<d_k(x)|\psi\>\d\mu(x),\hspace{0.5cm}\fii,\,\psi\in V_\e,
$$
(a minimal diagonalization of $\Mo$). 
\item $\Mo$ is normalized if and only if $\{\psi_m\}_{m=1}^{\dim\hi}$ is an ON set of $\hd$. Then $Y$ is an isometry.
\item $\Mo$ is a spectral measure if and only if $\{\psi_m\}_{m=1}^{\dim\hi}$ is an ON basis of $\hd$. Then $Y$ is a unitary operator and $\hd$ can be identified with $\hi$.
\item Let $\Mo\in\O(\Sigma,\hi)$. Then $\Mo\in\ext\O(\Sigma,\hi)$ if and only if, for any decomposable operator $D=\int_\Omega^\oplus D(x)\d\mu(x)\in\mathcal L(\hd)$, the condition $Y^*DY=0$ implies $D=0$.
\end{enumerate}
\end{theorem}
By using the embedding $\hd\subseteq L^2(\mu,\li)$ and an ON basis $\bf b$ of $\li$, the relation between vectors $\psi_n(x)$ and $d_k(x)$ can be chosen to be $\<\psi_n(x)|b_k\>=\<h_n|d_k(x)\>$ so that we may then write
$$
\<\fii|\Mo(X)\psi\>=\int_X \<\fii| \A^1(x)^* \A^1(x)\psi\>\d\mu(x)=\int_X \sum_{k=1}^{n(x)} \<\fii|\A_k(x)^*\A_k(x)\psi\>\d\mu(x),\hspace{0.5cm}\fii,\,\psi\in V_\e,
$$
where $ \A^1(x):=\sum_{k=1}^{n(x)}\kb{b_k}{d_k(x)}=\sum_{n=1}^{\dim\hi}\kb{\psi_n(x)}{h_n}$ and $\A_k(x):=\kb1{d_k(x)}$ are (possibly unbounded) operators $V_\e\to\hi_{n(x)}$ and $V_\e\to\C$, respectively. In addition,
$$
 \A^1(x)^* \A^1(x)= \sum_{k=1}^{n(x)} \A_k(x)^*\A_k(x)= \sum_{k=1}^{n(x)}\kb{d_k(x)}{d_k(x)}=\sum_{n,m=1}^{\dim\hi}\<\psi_n(x)|\psi_m(x)\>\kb{h_n}{h_m} 
$$
is an operator $V_\e\to V_\e^\times$ (or a sesquilinear form $V_\e\times V_\e\to\C$).
Also one sees that 
$\Mo\in\ext\O(\Sigma,\hi)$ if and only if, for any decomposable operator $D=\int_\Omega^\oplus D(x)\d\mu(x)\in\mathcal L(\hd)$, the condition 
$$
\int_\Omega \<\fii| \A^1(x)^*D(x) \A^1(x)\psi\>\d\mu(x)=0,\hspace{0.5cm}\fii,\,\psi\in V_\e,
$$
implies $D=0$.

Let $\M^\Mo:\,\Sigma\times\C\to\lh$ be the trivial instrument associated with a POVM $\Mo$. Then, for example,
$$
\<\fii|\M^\Mo(X,c)\psi\>=\int_X \sum_{k=1}^{n(x)} \<\fii|\A_k(x)^*c\A_k(x)\psi\>\d\mu(x),\hspace{0.5cm}\fii,\,\psi\in V_\e,\;X\in\Sigma,\;c\in\C,
$$
so that we have obtained a minimal pointwise Kraus form for $\M^\Mo$. The next theorem generalizes the above constructions to arbitrary instruments.

\subsection{Instruments}

Let $\M:\,\Sigma\times \lk\to \lh$ be an instrument and $\mu:\,\Sigma\to[0,\infty]$ a $\sigma$-finite measure such that $\Mo_{I_\ki}$ is absolutely continuous with respect to $\mu$.\footnote{ If the POVM $\Mo_{I_\ki}$ associated to an instrument $\M$ is absolutely continuous with respect to a $\sigma$-finite nonnegative measure $\mu$ then all operator measures $\Mo_B$, $B\in\lk$, are absolutely continuous with respect to $\mu$ (since $\M(X,B)\le\M(X,\|B\|I_{\ki})$ by positivity).}
Let $\e=\{h_n\}$ [resp.\ $\k=\{k_s\}$] be an ON basis of $\hi$ [resp.\ $\ki$]. Denote
$B_{st}:=\<k_s|B k_t\>$ for all $B\in\lk$ and $1\le s,\,t<\dim\ki+1$. 
\begin{theorem}\label{th2}
The exists  a direct integral $\hd=\int_\Omega^\oplus\hi_{n(x)}\d\mu(x)$
(with $n(x)\le\dim\hi\dim\ki$) 
such that, for all $X,\,X'\in\Sigma$ and $B,\,B'\in\lk$,
\begin{enumerate}
\item
$\M(X,B)=Y^*(B\otimes\CHII X)Y$ where $Y:\,\hi\to\ki\otimes\hd$, 
$$Y=\sum_{m=1}^{\dim\hi}\sum_{t=1}^{\dim\ki}\kb{k_t\otimes\psi^t_m}{h_m},$$ is an isometry (i.e.\ $\sum_{t=1}^{\dim\ki}\<\psi^t_n|\psi^t_m\>=\delta_{nm}$) and
$\{\psi^t_m\}_{m,t}\subseteq\hd$ is such that
the set of linear combinations of vectors $\CHI X \psi^t_m$, $X\in\Sigma$, $m<\dim\hi+1$, $t<\dim\ki+1$, is dense in $\hd$
 (a minimal Stinespring dilation for $\M$). Hence, by defining 
 $Y(X,B):=(B\otimes\CHII X) Y=\sum_{m=1}^{\dim\hi}\sum_{t=1}^{\dim\ki}\kb{(Bk_t)\otimes(\CHI X\psi^t_m)}{h_m}$,
\begin{eqnarray*}
\M(X\cap X',B^*B')&=&Y(X,B)^*Y(X',B') \\
&=&\sum_{n,m=1}^{\dim\hi}\sum_{s,t=1}^{\dim\ki}(B^*B')_{st}\int_{X\cap X'}\<\psi^s_n(x)|\psi^t_m(x)\>\d\mu(x)\kb{h_n}{h_m} \\
\end{eqnarray*}
weakly (a minimal Kolmogorov decomposition for $\M$). 
\item
The are measurable maps $d^t_k:\,\Omega\to V_\e^\times$ such that, for all $x\in\Omega$, the operators 
$$\A_k(x):=\sum_{t=1}^{\dim\ki}\kb{k_t}{d^t_k(x)},\hspace{0.5cm}1\le k<n(x)+1,$$ 
from $V_\e$ to $\ki$ are linearly independent, and
\begin{eqnarray*}
\<\fii|\M(X,B)\psi\>&=&\sum_{s,t=1}^{\dim\ki} B_{st}\int_X \sum_{k=1}^{n(x)} \<\fii|d^s_k(x)\>\<d^t_k(x)|\psi\>\d\mu(x) \\
&=& \int_X \sum_{k=1}^{n(x)} \<\fii|\A_k(x)^*B\A_k(x)\psi\>\d\mu(x),
\hspace{0.5cm}\fii,\,\psi\in V_\e,
\end{eqnarray*}
(a minimal pointwise Kraus form of $\M$).
By defining operators
$$
\A^t(x):=\sum_{k=1}^{n(x)} \kb{b_k}{d_k^t(x)},%=\sum_{n=1}^{\dim\hi}\kb{\psi_n^s(x)}{h_n}
%=\sum_{k=1}^{n(x)}\kb{b_k}{k_s}\A_k(x)
\hspace{0.5cm}
1\le t<\dim\ki+1,\;\;x\in\Omega,
$$
from $V_\e$ to $\hi_{n(x)}$ one gets
$$
\<\fii|\M(X,B)\psi\>=\sum_{s,t=1}^{\dim\ki} B_{st} \int_X \<\fii|\A^s(x)^*\A^t(x)\psi\>\d\mu(x),
\hspace{0.5cm}\fii,\,\psi\in V_\e.
$$
\item 
$\M(X,B)=\sum_{k=1}^{n(X)} \A_k(X)^*B\A_k(X)$ ultraweakly
 where $n(X)\le\dim\hi\dim\ki$ and, for all $X\in\Sigma$, bounded operators $\A_k(X):\,\hi\to\ki$, $1\le k<n(X)+1$, are linearly independent  %such that $\sum_{k=1}^{n(X)} \A_k(X)^*\A_k(X)\le I_\hi$
(a minimal setwise Kraus form of $\M$).
\item
$\M\in\ext\In(\Sigma,\ki,\hi)$ if and only if, for any decomposable operator $D=\int_\Omega^\oplus D(x)\d\mu(x)\in\mathcal L(\hd)$, the condition $Y^*(I_{\ki}\otimes D)Y=0$ implies $D=0$.
\item
$\M\in\ext\In(\Sigma,\ki,\hi)$ if and only if, for any decomposable operator $D=\int_\Omega^\oplus D(x)\d\mu(x)\in\mathcal L(\hd)$, the condition 
$$
\sum_{s=1}^{\dim\ki} \int_\Omega\<\fii|\A^s(x)^*D(x)\A^s(x)\psi\>\d\mu(x)=0,
\hspace{0.5cm}\fii,\,\psi\in V_\e,
$$ 
implies $D=0$.
\end{enumerate}
\end{theorem}

\begin{proof}
(1)
Let $\M(X,B)=W^*\big(B\otimes\sfe(X)\big)W$ be a Stinespring dilation of $\M$, where 
$W:\,\hi\to\ki\otimes\hi'$ is an isometry and $\sfe:\,\Sigma\to\mathcal L(\hi')$ a spectral measure acting on a possibly {\it nonseparable} Hilbert space $\hi'$ (see, e.g.\ \cite{Da,Ho98,Oz84,St}). Write $W$ of the form
$$
W=\sum_{m=1}^{\dim\hi}\sum_{t=1}^{\dim\ki}\kb{k_t\otimes\fii^t_m}{h_m},
$$
where $\fii^t_m\in\hi'$ and $\sum_{t=1}^{\dim\ki}\<\fii^t_n|\fii^t_m\>=\delta_{nm}$, and let $\hi''\subseteq\hi'$ be a {\it separable} Hilbert space spanned by countable set of vectors $\fii^t_m$. Denote by $P$ the projection from $\hi'$ onto $\hi''$.
Then $$\M(X,B)=W^*\big(B\otimes P\sfe(X)P\big)W$$ 
where $X\mapsto P\sfe(X)P$ can be viewed as a POVM $\Sigma\to\mathcal L(\hi'')$, and hence it can be diagonalized by Theorem \ref{th1}. We have 
$\M(X,B)=  Y^*(B\otimes\CHII X)  Y$ where $  Y:\,\hi\to\ki\otimes\hi_\oplus'$, 
$$  Y=\sum_{m=1}^{\dim\hi}\sum_{t=1}^{\dim\ki}\kb{k_t\otimes  \psi^t_m}{h_m},$$ is an isometry and $\hi_\oplus'=\int_\Omega^\oplus\hi_{n'(x)}\d\mu(x)\subseteq L^2(\mu,\li)$ a direct integral Hilbert space. 

Let then $\hd\subseteq\hi_\oplus'$ be a closure of the linear span of vectors $\CHI X  \psi^t_m$.
Especially, any $\psi^t_m\in\hd$.
Since 
$\hd$ is a Hilbert (sub)space we may define a projection $R$ from $\hi_\oplus'$ onto $\hd$. It is easy to see that $R$ commutes with any $\CHII X$ so that it is decomposable, $R=\int_\Omega^\oplus R(x)\d\mu(x)$, by Proposition \ref{prop1}.
Hence, $\hd$ is a direct integral and unitarily equivalent to some $\int_\Omega^\oplus\hi_{n(x)}\d\mu(x)\subseteq L^2(\mu,\li)$ (defined as before) so that we may set
$\hd=\int_\Omega^\oplus\hi_{n(x)}\d\mu(x)$.
Now $Y$ is actually an isometry from $\hi$ to $\ki\otimes\hd$.
Note that (almost everywhere) $n(x)\le n'(x)\le\dim\hi''\le\dim\hi\dim\ki$.
Obviously vectors $(B\otimes\CHII X) Y\psi$, $B\in\lk$, $X\in\Sigma$, $\psi\in\hi$, span
$\ki\otimes\hd$ so that the above Stinespring dilation is minimal.
Since 
\begin{eqnarray*}
\<h_n|\M(X,B)h_m\> &=& \<h_n|Y^*(B\otimes\CHII X)  Yh_m\> =\sum_{s,t=1}^{\dim\ki}B_{st}\<\psi^s_n|\CHI X\psi^t_m\>\\
&=&\sum_{s,t=1}^{\dim\ki}B_{st}\int_{X}\<\psi^s_n(x)|\psi^t_m(x)\>\d\mu(x)
\end{eqnarray*}
and the first part of the proof follows.

(2) Pick a representative $\Omega\ni x\mapsto\psi_n^s(x)\in\hi_{n(x)}$ from any class $\psi_n^s$ 
such that $$\sum_{s=1}^{\dim\ki}\|\psi_n^s(x)\|^2<\infty$$ for {\it all} $x\in\Omega$ and $n<\dim\hi+1$. This is possible since the set $\{\psi_n^s\}_{n,s}$ is countable and
$$\<h_n|\M(X,I_\ki)h_n\>=\sum_{s=1}^{\dim\ki}\int_{X}\|\psi_n^s(x)\|^2\d\mu(x)\le1.$$
Then, for all $x\in\Omega$, define $d_k^s(x):=\sum_{n=1}^{\dim\hi}\<\psi_n^s(x)|b_k\>h_n\in V_\e^\times$ and  
$
\A_k(x):=\sum_{s=1}^{\dim\ki}\kb{k_s}{d^s_k(x)}
$
so that
$\A_k(x)h_n=\sum_{s=1}^{\dim\ki}\<b_k|\psi_n^s(x)\>{k_s}$
implying
\begin{eqnarray*}
\|\A_k(x)h_n\|^2&=&\sum_{s=1}^{\dim\ki}|\<b_k|\psi_n^s(x)\>|^2 \\
&\le&\sum_{k=1}^{n(x)}\sum_{s=1}^{\dim\ki}|\<b_k|\psi_n^s(x)\>|^2=\sum_{s=1}^{\dim\ki}\|\psi_n^s(x)\|^2<\infty
\end{eqnarray*}
and thus $\A_k(x)\psi\in\ki$ for all $\psi\in V_\e$.
Moreover,
\begin{eqnarray*}
\<\psi^s_n(x)|\psi^t_m(x)\>&=&\sum_{k=1}^{n(x)}\<\psi_n^s(x)|b_k\>\<b_k|\psi_m^t(x)\>
=\sum_{k=1}^{n(x)} \<h_n|d^s_k(x)\>\<d^t_k(x)|h_m\> \\
&=&\sum_{k=1}^{n(x)}\<h_n|\A_k(x)^*k_s\>\<k_t|\A_k(x)h_m\>.
\end{eqnarray*}
For all $x\in\Omega$, the operators $\A_k(x)$, $k<n(x)+1$, can be chosen to be linearly independent:
Indeed, suppose that the exists a set $X'\in\Sigma$ such that, for all $x\in X'$,  $n(x)>0$ and the set $\{\A_k(x)\}_{k=1}^{n(x)}$ is linearly dependent. Then, for all $x\in X'$, there exists complex numbers $c^x_k$, $k<n(x)+1$, such that $c^x_k\ne 0$ for finitely many $k$'s and $\sum_{k=1}^{n(x)}c_k^x\A_k(x)=0$. This implies that, by defining a nonzero $\fii_x:=\sum_{k=1}^{n(x)}\overline{c_k^x} b_k\in\hi_{n(x)}$, 
$$
\<\fii_x|\psi_n^s(x)\>=\sum_{k=1}^{n(x)}c_k^x\<b_k|\psi_n^s(x)\>= \sum_{k=1}^{n(x)}c_k^x\<k_s|\A_k(x)h_n\>=0
$$
for all $x\in X'$, $n<\dim\hi+1$, and $s<\dim\ki+1$. 
Let $\hi_x$ be the closure of $${{\rm lin}\{\psi_n^s(x)\,|\,n<\dim\hi+1,\,s<\dim\ki+1\}}$$ in $\hi_{n(x)}.$
Since, for all $x\in X'$,  the orthogonal complement $\hi_x^\perp$ of $\hi_x$ in $\hi_{n(x)}$ is nonzero, we may choose a $\fii\in\hd$ such that $0\ne \fii(x)\in\hi_x^\perp$ for all $x\in X'$.
But then $\<\CHI{X'}\fii|\CHI X\psi_n^s\>=0$ for all $X,\,n,\,s$ implying that $\CHI{X'}\fii=0$ by the density of the linear combinations of the vectors $\CHI X\psi_n^s$. Hence, $\mu(X')=0$ and we may simply redefine $n(x)$ to be zero for all $x\in X'$.

Finally, $\A^t(x):=\sum_{k=1}^{n(x)} \kb{b_k}{d_k^t(x)}=\sum_{n=1}^{\dim\hi}\kb{\psi_n^t(x)}{h_n}$ 
is obviously an operator from $V_\e$ to $\hi_{n(x)}$ and the last claim follows easily.

(3) The equation $\M(X,B)=\sum_{k=1}^{n(X)} \A_k(X)^*B\A_k(X)$ is just the usual Kraus decomposition \cite{Kr} of a completely positive map $B\mapsto\M(X,B)$.

(4) From Theorem \ref{seiska} and Example \ref{instruex} we see that, since 
the unital ${}^*$-homomorphism is now
$$
\pi:\,\lk\otimes L^\infty(\mu)\to\li(\ki\otimes\hd),\;B\otimes f\mapsto B\otimes\hat f,
$$ 
the instrument $\M$ is extreme if and only if, for all $F\in\mathcal L(\ki\otimes\hd)$ such that $F\pi(B\otimes f)=\pi(B\otimes f)F$ for all  $B\in\lk$, $f\in L^\infty(\mu)$, the condition
$Y^*FY=0$ implies $F=0$.
But if $F(B\otimes\hat f)=(B\otimes\hat f)F$ for all $B$ and $f$ then $F=I_\ki \otimes D$ for some decomposable $D\in\mathcal L(\hd)$ (since the commutant of the tensor product of two von Neumann algebras is the tensor product of the communtants of the algebras in question; the commutant of $\lk$ is $\C I_\ki$, see also Proposition  \ref{prop1}).
Finally, (5) follows immediately from (4).
\end{proof}

We have collected the basic operators and vectors related to an instrument in the Appendix, see Remark \ref{remu}. Moreover, we show there that the results of Theorem \ref{th2} do not essentially depend, e.g., on the choices of the bases $\bf h$ and $\bf k$.

\begin{remark}\rm\label{exrem}
The extremality condition $\sum_{s=1}^{\dim\ki} \int_\Omega\<\fii|\A^s(x)^*D(x)\A^s(x)\psi\>\d\mu(x)=0,$ $\fii,\,\psi\in V_\e,$ of item (5) of Theorem \ref{th2} can equivalently be written in the following forms:
\begin{eqnarray*}
&&\sum_{s=1}^{\dim\ki} \int_\Omega\<\psi_n^s(x)|D(x)\psi_m^s(x)\>\d\mu(x)=0, \hspace{0.5cm}1\le n,\,m<\dim\hi+1,\\
&& \sum_{s=1}^{\dim\ki}\int_\Omega\sum_{k,l=1}^{n(x)}D(x)_{kl}\<\fii|d_k^s(x)\>\<d_l^s(x)|\psi\>\d\mu(x)=0,
\hspace{0.5cm} \fii,\,\psi\in V_\e,\\
&&\int_\Omega\sum_{k,l=1}^{n(x)}D(x)_{kl}\<\fii|\A_k(x)^*\A_l(x)|\psi\>\d\mu(x)=0,
\hspace{0.5cm} \fii,\,\psi\in V_\e,
\end{eqnarray*}
where $D(x)_{kl}:=\<{b_k}|D(x){b_l}\>$.
A mathematically elegant characterization of extreme instruments is the following:
Let $\mathcal D(\hd)\subseteq\mathcal L(\hd)$ be the $C^*$-algebra of decomposable operators $D=\int_\Omega^\oplus D(x)\d\mu(x)\in\mathcal L(\hd)$. It is the commutant of the $C^*$-algebra $L^\infty(\mu)\subseteq\mathcal L(\hd)$ by Proposition \ref{prop1}.
For any $\M\in\In(\Sigma,\ki,\hi)$ (with $\hd$ and $Y$ as in Theorem \ref{th2}) define its bilinear `extension' 
$$
\ov\M:\,\mathcal D(\hd)\times\lk\to\lh,\;(D,B)\mapsto\ov\M(D,B):=Y^*(B\otimes D)Y
$$
for which $\ov\M(\CHII X,B)=\M(X,B)$.
Then, $\M\in\ext\In(\Sigma,\ki,\hi)$ if and only if $D\mapsto\ov\M(D,I_\ki)$ is injective.
\end{remark}

\begin{remark}[Dirac formalism]\rm\label{dirac}
Let $S=S^*$ be a (possibly unbounded) self-adjoint operator on $\hi$ and $\mathsf M$ its spectral measure (defined on the Borel $\sigma$-algebra of $\R$). Let $d_k(x)$ be the generalized vectors of Theorem \ref{th1} associated with $\Mo$.
As shown in \cite{HyPeYl}, there exists an ON basis $\e$ of $\hi$ such that $SV_\e\subseteq V_\e$,
and if $S^\times:\,V_\e^\times\to V_\e^\times$ is the extension of $S$, one gets $S^\times d_k(x)=xd_k(x)$ for almost all $x$ in the spectrum of $S$. Hence, Theorem \ref{th1} can be viewed as a generalization of {\it Dirac formalism} for POVMs and we may call $n(x)$ the multiplicity of a measurement outcome $x$.

Solutions $d_k(x)$ of the `eigenvalue' equation $S^\times d_k(x)=xd_k(x)$ turn out to be extremely useful for determining the spectral measure $\Mo$ of $S$ in many practical situations.
Similarly, the generalized vectors $d_k^t(x)$ of
item (2) of Theorem \ref{th2} are useful, e.g., for determining $\Mo$-compatible instruments $\M$ as we will see later.
For example, if vectors $d_k^t(x)$ are related with $\M$ then the generalized vectors $d_k(x)$ of the associate POVM $\Mo$ of $\M$ must satisfy
$$
\sum_{k=1}^{n'(x)}\kb{d_k(x)}{d_k(x)}=
\sum_{t=1}^{\dim\ki}\sum_{k=1}^{n(x)}\kb{d^t_k(x)}{d^t_k(x)}
$$
(weakly on $V_\e$) where the multiplicities $n(x)$ and $n'(x)$ are not necessarily the same.
\end{remark}

\begin{remark}\rm
It should be stressed that the direct integral Hilbert space $\hd$ of Theorem \ref{th2} is not necessarily separable. However, when $\Sigma$ is countably generated then $\hd$ is separable (see, e.g.\ \cite{Ca}).
This holds, for instance, when $\Sigma$ is the Borel $\sigma$-algebra $\bo\Omega$ of a second countably topological space $\Omega$. Hence, in the physically relevant examples, one can assume that $\hd$ is separable.
Note that in the proof of (1) of Theorem \ref{th2}, one cannot assume that the spectral measure $\sfe$ of a Stinespring dilation acts on a separable Hilbert space. Thus, in nonseparable cases, it is questionable whether one can directly diagonalize $\sfe$.
\end{remark}

\subsection{Examples}\label{examples}

In this subsection, we consider some special cases of Theorem \ref{th2}.
We assume that $\M\in\In(\Sigma,\ki,\hi)$ and use the notations of Theorem \ref{th2}.

\begin{example}[Observables: $\dim\ki=1$] \rm \label{ExPOVM}
Applying Theorem \ref{th2} to the trivial instrument $\M=\M^\Mo$ of a POVM $\Mo$,
we see that Theorem \ref{th1} is a special case of Theorem \ref{th2}.
Indeed, now $\ki=\C$, $\ki\otimes\hd\cong\hd$, indices $t$ and $s$ run from 1 to $\dim\ki=1$, $k_1=1$, and $B_{11}=c\in\C$. It is easy to see that $\Mo\mapsto \M^\Mo$ is an affine bijection from $\O(\Sigma,\hi)$ onto
$\In(\Sigma,\C,\hi)$, and that $\M^\Mo\in\ext\In(\Sigma,\C,\hi)$ if and only if $\Mo\in\ext\O(\Sigma,\hi)$.
\end{example}

\begin{example}[Preparations: $\dim\hi=1$] \rm
Let $\hi=\C$ (and $\lh\cong\C$) so that indices $m$ and $n$ run from $1$ to $1$ and $h_1=1$, and let $\M\in\In(\Sigma,\ki,\C)$.
Drop indices $n$ and $m$ out from the notations.
Then one has the following identifications: $\psi^s:=\psi^s_1$, $d^s_k(x)=\<\psi^s(x)|b_k\>\in\C$, $\A_k(x)=a_k(x):=\sum_s\overline{d^s_k(x)}k_s\in\ki$ (linearly independent vectors), $\A^s(x)=\psi^s(x)$, and
$$
\M(X,B)=\sum_{s,t=1}^{\dim\ki}B_{st}\int_X\<\psi^s(x)|\psi^t(x)\>\d\mu(x)=\int_X\sum_{k=1}^{n(x)}\<a_k(x)|Ba_k(x)\>\d\mu(x)=\tr{\rho(X)B}
$$
where 
$$
\rho:\,\Sigma\to\mathcal T(\ki),\;X\mapsto \rho(X):=\int_X\sum_{k=1}^{n(x)}\kb{a_k(x)}{a_k(x)}\d\mu(x)
$$
is a positive trace class valued operator measure for which
$$
1=\tr{\rho(\Omega)}=\sum_{s=1}^{\dim\ki}\int_\Omega\|\psi^s(x)\|^2\d\mu(x)=\int_\Omega\sum_{k=1}^{n(x)}\|a_k(x)\|^2\d\mu(x),
$$
i.e.\ $\rho(\Omega)\in\mathcal S(\ki)$ is a state. By defining linearly independent unit vectors $\fii_k(x):=a_k(x)/\|a_k(x)\|$ (when $a_k(x)\ne0$) 
and $\lambda_k(x):=\|a_k(x)\|^2\in[0,1]$ one sees that 
$$
\rho(\Omega)=\int_\Omega\sum_{k=1}^{n(x)}\lambda_k(x)\kb{\fii_k(x)}{\fii_k(x)}\d\mu(x)
$$
is a {\it (possibly uncountable or continuous) `convex combination' of pure states $\kb{\fii_k(x)}{\fii_k(x)}$.}
Physically, $\rho(\Omega)$ could be associated with some preparation procedure which produces convex combinations of pure states. For example, a radiation source emits pure states $\kb{\fii_k(x)}{\fii_k(x)}$ randomly so that the output state must be assumed to be mixed with weights $\lambda_k(x)$.

From Theorem \ref{th1} one sees that, any positive operator measure $\Mo:\,\Sigma\to\lk$ for which $\Mo(\Omega)\in\mathcal T(\ki)$ and $\tr{\Mo(\Omega)}=1$, 
is of the form
$$
\Mo(X)=\sum_{t,s=1}^{\dim\ki}\int_X\<\psi_t(x)|\psi_s(x)\>\d\mu(x)\kb{k_t}{k_s}\in\mathcal T(\ki)
$$
where $\tr{\Mo(\Omega)}=\sum_{s=1}^{\dim\ki}\int_\Omega\|\psi_s(x)\|^2\d\mu(x)=1$. 
Define then a positive operator measure $\rho':\,\Sigma\to\mathcal T(\ki)$ by
$$
\rho'(X):=\sum_{t,s=1}^{\dim\ki}\int_X\<\psi_s(x)|\psi_t(x)\>\d\mu(x)\kb{k_t}{k_s},
$$
and an instrument 
$\M'\in\In(\Sigma,\ki,\C)$ by $\M'(X,B):=\tr{\rho'(X)B}$. The correspondence $\Mo\mapsto\M'$ is an affine bijection.

Note that $\M\in\ext\In(\Sigma,\ki,\C)$ if and only if $\sum_{s=1}^{\dim\ki}\int_\Omega\<\psi^s(x)|D(x)\psi^s(x)\>\d\mu(x)$
implies $D=0$. We have two special cases:

\noindent
a) Also $\ki=\C$. Then $\In(\Sigma,\C,\C)$ is just a convex set of probability measures (classical states) $\mu:\,\Sigma\to[0,1]$ (for which $\psi^1(x)\equiv1$) and we see that $\mu\in\ext\In(\Sigma,\C,\C)$ if and only if, for all $d(x)\in L^\infty(\mu)$, the condition $\int_\Omega d(x)\d\mu(x)=0$ implies $d=0$, if and only if $\mu(\Sigma)\in\{0,1\}$.

\noindent
b) Let $\Omega=\{0\}$ and consider $\M\in\In(2^{\{0\}},\ki,\C)$. Then the corresponding $\rho(\{0\})\in\mathcal T(\ki)$ is a state and $\In(2^{\{0\}},\ki,\C)$ is a convex set of states on $\ki$. One sees immediately that $\M$ is extreme if and only if $\rho(\{0\})$ is pure, i.e.\ $\rho(\{0\})=\kb{a_1(\{0\})}{a_1(\{0\})}$.
\end{example}

\begin{example}[The discrete case]\rm\label{ChEx}
Let $\M\in\In(\Sigma,\ki,\hi)$ and $\Mo_{I_\ki}$ the corresponding POVM.
Suppose that there exists a finite or countably infinite set $\ov X=\{x_i\}_{i=1}^N\subseteq\Omega$, $N\in\N_\infty$, such that $\{x_i\}\in\Sigma$, $\Mo_{I_\ki}(\{x_i\})>0$ for all $i$ and $\Mo_{I_\ki}(\ov X)=I_{\hi}$.
Then $\mu$ can be chosen to be such that $\mu(\{x_i\})=1$ for all $i$ and $\mu(\Omega\setminus\ov X)=0$.
It follows that 
$$
\M(X,B)=\sum_{1\le i<N+1\atop x_i\in X}T_i(B)
$$
where any $T_i:\,\lk\to\lh,\,B\mapsto T_i(B):=\M(\{x_i\},B)$ is a CP operation and all integrals in Theorem \ref{th2} reduce to sums. By replacing each $x_i$ by $i$ in the notations [e.g.\ $\A_k(i):=\A_k(x_i)=\A_k(\{x_i\})$] we have 
$\hd=\bigoplus_{i=1}^N\hi_{n(i)}$, $\psi_m^t(i)\in\hi_{n(i)}$, $d_k^t(i)\in\hi$, and the extended operators $\A_k(i):\,\hi\to\ki$ and $\A^t(i):\,\hi\to\hi_{n(i)}$ are bounded. Hence, we have $\sum_{t,i}\<\psi_n^t(i)|\psi_m^t(i)\>=\delta_{nm}$ and
(weakly)
\begin{eqnarray*}
T_i(B)&=&\sum_{n,m=1}^{\dim\hi}\sum_{s,t=1}^{\dim\ki}B_{st}\<\psi_n^s(i)|\psi_m^t(i)\>\kb{h_n}{h_m}=\sum_{s,t=1}^{\dim\ki}B_{st}\sum_{k=1}^{n(i)}\kb{d_k^s(i)}{d_k^t(i)}\\
&=&\sum_{k=1}^{n(i)}\A_k(i)^*B\A_k(i)=\sum_{s,t=1}^{\dim\ki}B_{st}\A^s(i)^*\A^t(i)
\end{eqnarray*}
weakly. We say that $n(i)$ is the {\it rank} of $T_i$.
Immediately one gets a generalization of Theorem 5 of \cite{DAPeSe}:
 $\M\in\ext\In(\Sigma,\ki,\hi)$ if and only if $\sum_{i,k,l}D(i)_{kl}\A_k(i)^*\A_l(i)=0$ implies $D=0$ (where $D=\bigoplus_{i=1}^N D(i)$ is bounded).
In the next example, we concentrate on the channel case $N=1$ and drop $(i)$ out from the notations.
\end{example}

\begin{example}[Channels]\rm
Let $T:\,\lk\to\lh$ be a CP channel and 
$T(B)=\sum_{k=1}^n\A_k^*B\A_k$, $B\in\lk$, its minimal Kraus decomposition (where the bounded operators $\A_k:\,\hi\to\ki$ are linearly independent and $n\le\dim\hi\dim\ki$ is minimal).
From Example \ref{ChEx} we get %\footnote{Obviously, Theorem \ref{C4} works also in a slighty more general case where one consideres a convex set of CP-maps $T:\,\lk\to\lh$ where $T(I_\ki)=K$ is fixed but not necessarily $I_\hi$.} 
that the channel $T$ is an extreme point of the convex set of CP channels $\lk\to\lh$ if and only if, for any $n\times n$--complex matrix $(D_{kl})$ with the finite operator norm the condition
$
\sum_{k,l=1}^nD_{kl}\A_k^*\A_l=0
$
(weakly) implies $D_{kl}\equiv0$.
This has also been proved in \cite[Theorem 2.4]{Ts}
(see also \cite[Theorem 5]{Ch} and an alternative formulation \cite[Proposition 1]{Ho11}).
Especially,
when the rank $n$ of $T$ is finite, then $T$ is extreme if and only if the operators $\A_k^*\A_l$ are linearly independent. %\footnote{If $n=\infty$ and $T$ is extreme then operators $\A_k^*\A_l$ are linearly independent.}
For example, if the rank $n=1$, i.e.\ $T(B)=\A_1^*B\A_1$, then $T$ is extreme.

\end{example}

\begin{example}[The finite dimensional case] \rm

Let $\M\in\ext\In(\Sigma,\ki,\hi)$ with the isometry $Y$ of Theorem \ref{th2}, and let $\Mo$ be the associate POVM of $\M$. For any $f\in L^\infty(\mu)$, the condition
\begin{equation}
\label{fjenhvudjfv}
\int_\Omega f(x)\d\Mo(x)=Y^*(I_{\ki}\otimes \hat f)Y=0
\end{equation} 
implies $f=0$.
Let $X_1,\ldots,X_N\in\Sigma$ be disjoint sets such that $\Mo(X_i)\neq0$. Then the effects $\Mo(X_1),\ldots,\Mo(X_N)$ are linearly independent which can be seen by substituting
$f=\sum_{i=1}^N c_i\CHII{X_i}$ into equation \eqref{fjenhvudjfv}:
$$
\sum_{i=1}^N c_i\Mo(X_i)=0\qquad\text{implies}\qquad
c_1=c_2=\ldots=c_N=0.
$$
From this fact follows that, if $\M$ is extremal then there are at most $(\dim\hi)^2$ disjoint sets $X_i$ such that $\Mo(X_i)\ne 0$. Since $\Mo(X)=0$ implies $\M(X,B)=0$ for all $B\in\lk$, an extremal $\M$ is concentrated on the set $\cup_{i=1}^N X_i$, $N\le(\dim\hi)^2$, but it does not necessarily follow that $\Mo$ is discrete even when $\dim\hi<\infty$ \cite{HaHePe}. By adding topological assumptions we get (for the proof, see \cite{HaHePe}):

\begin{proposition}\label{lksjdf}
Suppose that $(\Omega,\Sigma=\bo\Omega)$ is a second countable Hausdorff space, $\dim\hi <\infty$, and $\M\in\ext\In(\Sigma,\ki,\hi)$.
Then $\M$ is concentrated on a finite set, i.e.,
\begin{equation*}
\M(X,B)=\sum_{i=1}^N \CHI X(x_i) T_i(B),\qquad X\in\Sigma,\;B\in\lk,
\end{equation*}
for some finite number $N\leq(\dim\hi)^2$ of elements $x_1,\ldots,x_N\in\Omega$ and CP operations $T_i:\,\lk\to\lh$.
\end{proposition}

\end{example}

\begin{example}[Multi-instruments] \rm
Quantum measurement and information processes can be viewed as combinations of the basic building blocks, elementary instruments, introduced in the preceding examples. For example, first one could start from a trivial (or `classical') Hilbert space $\C$ and prepare a state in $\ki$. Then one could process the state by using channels and instruments. Finally, the process ends in a measurement of a POVM, the final Hilbert space being the classical space $\C$ again. This kind of processes can be viewed as the following combinations of (Heisenberg) instruments
$\M_i\in\In(\Sigma_i,\ki_i,\hi_i)$, $i=1,2,\ldots,N\in\N_+$,
$$
\M_1\bigg(X_1,\M_2\Big(X_2,\M_3\big(X_3,\cdots\M_N(X_N,B_N)
\big)
\Big)
\bigg)\in\li(\hi_1),\qquad X_i\in\Sigma_i,\;B_N\in\li(\ki_N)
$$
so that one must have 
$\ki_i=\hi_{i+1}$ for all $i=1,\ldots,N-1$.
We denote the above combination by
$
\ov\M\big((X_1,X_2,\ldots,X_N),B_N\big)
$
and say that the corresponding map $$\ov\M:\,\Sigma_1\times\Sigma_2\times\cdots\times\Sigma_N\times\li(\ki_N)\to\li(\hi_1)$$ is the {\it multi-instrument} generated by the instruments $\M_1,\ldots,\M_N$.
By induction, it suffices to consider the case $N=2$.
Next we consider an important example (see a recent paper \cite{CaHeTo}).

Assume that we measure POVMs $\Mo_1$ and $\Mo_2$ (of the same Hilbert space $\hi$) by performing their measurements sequentially (first $\Mo_1$ and then $\Mo_2$). This leads to the bi-instrument $\ov\M$ defined by 
$$
\ov\M(X_1\times X_2,B_2):=\M_1\big(X_1,\M_2(X_2,B_2)\big).
$$ 
Usually $\ov\M$ extends to an instrument on a product $\sigma$-algebra \cite[Theorem 2]{DaLe}.
It defines a sequential joint observable $\Mo_{12}$ whose margins are POVMs 
$$
X_1\mapsto\Mo_{12}(X_1\times\Omega_2,I_\hi)=\Mo_1(X_1),\qquad X_2\mapsto\Mo_{12}(\Omega_1\times X_2,I_\hi)=\M_1\big(\Omega_1,\Mo_2(X_2)\big)
$$ 
where the channel $\M_1\big(\Omega_1,\bullet)$ operates to $\Mo_2$, that is, the first measurement disturbs the subsequent one.
Finally, we note that, if $\{\psi_n^s\}$ and 
$\{ {\psi'}_{n}^{s} \}$ are the structure vectors of $\M_1$ and $\M_2$ of Theorem \ref{th2}, respectively, the (not necessarily minimal) structure vectors $\{{\psi''}_n^s\}$ of $\ov\M$ can be easily calculated:
$$
{\psi''}_n^s=\sum_{a} {\psi}_n^a\otimes{\psi'}_a^s.
$$
\end{example}

\section{$\Mo$-compatible instruments}\label{luku4}

Let $\M\in\In(\Sigma,\hi,\ki)$ with an isometry 
$Y:\,\hi\to\ki\otimes\hd$ and structure vectors $\psi_n^s\in  L^2(\mu,\li)$ given by (1) of Theorem \ref{th2}. Let $\{\Mo_B\}_{B\in\lk}$ be the family of operator measures associated to $\M$. Since any $B$ can be decomposed into positive parts one sees that $\Mo_B$ can also be decomposed into a sum of positive operator measures (compare to \cite{HyPeYl2}).
Suppose then that $B$ is positive with the square root operator $\sqrt B$.
Now
$$
\Mo_B(X)=\M(X,B)=Y^*(B\otimes\CHII X)Y=Y_B^*(I_\ki\otimes\CHII X)Y_B
$$
where
$$
Y_B:=\big(\sqrt{B}\otimes I_{\hd}\big)Y
$$
is an operator from $\hi$ to a Hilbert space
$\hi^B\subseteq\ki\otimes\hd$ defined as the closure of the linear combinations of vectors 
$\big(\sqrt{B}\otimes\CHII X\big)Y\psi$, $X\in\Sigma$, $\psi\in\hi$.
Trivially, vectors $(I_\ki\otimes\CHII X)Y_B\psi$ span $\hi^B$
so that we have obtained a minimal Naimark dilation for $\Mo_B$.
(By Theorem  \ref{th1}, $\hi^B$ is unitarily equivalent with a direct integral Hilbert space $\hi_\oplus^B\subseteq L^2(\mu,\li)$ where $\Mo_B$ is diagonalized but we do not need this fact here.)
By defining the structure vectors of $\Mo_B$,
$$
\psi^B_n:=Y_B h_n=\big(\sqrt{B}\otimes I_{\hd}\big)Yh_n
=\sum_{s=1}^{\dim\ki}\big(\sqrt{B}k_s\big)\otimes\psi_n^s
$$
one gets
\begin{eqnarray*}
\Mo_B(X)&=&\sum_{n,m=1}^{\dim\hi}\int_{X}\<\psi^B_n(x)|\psi^B_m(x)\>\d\mu(x)\kb{h_n}{h_m} \\
&=&\sum_{n,m=1}^{\dim\hi}\sum_{s,t=1}^{\dim\ki}B_{st}\int_X\<\psi^s_n(x)|\psi^t_m(x)\>\d{\mu}(x)\kb{h_n}{h_m}=\M(X,B). 
\end{eqnarray*}
When $B=I_{\ki}$, we see that the structure vectors $\psi_n^{I_\ki}$ of the associate observable $\Mo_{I_\ki}$ of $\M$ are
$$
\psi^{I_\ki}_n=Yh_n=\sum_{s=1}^{\dim\ki}k_s\otimes\psi_n^s.
$$
If $B$ is not positive then one can collect the structure vectors of its positive parts into a single structure vector as in Remark 5.8 of \cite{HyPeYl2}.

Let then $\Mo\in\O(\Sigma,\hi)$ be a POVM with ON vectors $\psi_n\in \hd$ of Theorem \ref{th1}
and try to find 
an {\it $\Mo$-compatible instrument} 
$\M\in\In(\Sigma,\hi,\ki)$, i.e.\ $\M(X,I_\ki)\equiv\Mo(X)$, with vectors $\psi_n^s$ as in Theorem \ref{th2}. %(see, Remark \ref{remu} in the appendix).
Clearly, one must `solve' the equation
\begin{equation}\label{solve}
\sum_{s=1}^{\dim\ki}\<\psi^s_n(x)|\psi^s_m(x)\>\equiv\<\psi_n(x)|\psi_m(x)\>
\end{equation}
for unknown vectors $\psi_n^s$. The next theorem characterizes completely $\Mo$-compatible instruments:

\begin{theorem}\label{compa}
Let $\Mo\in\O(\Sigma,\hi)$ with the structure vectors $\psi_m\in\hd=\int_\Omega^\oplus\hi_{n(x)}\d\mu(x)$ and the isometry $Y$ of Theorem \ref{th1} associated with the diagonal minimal Naimark dilation, and let $\M\in\In(\Sigma,\ki,\hi)$. Then
$\M(X,I_\ki)\equiv\Mo(X)$ if and only if
there exists a decomposable CP channel $T:\,\lk\to\li(\hd)$,
$B\mapsto T(B)=\int_\Omega^\oplus T_x(B)\d\mu(x)$, where $T_x:\lk\to\li(\hi_{n(x)})$ are CP channels for $\mu$-almost all $x\in\Omega$, such that
$$
\M(X,B)\equiv Y^*T(B)\CHII X Y
=\sum_{n,m=1}^{\dim\hi}\int_{X}\<\psi_n(x)|T_x(B)\psi_m(x)\>\d\mu(x)\kb{h_n}{h_m}.
$$
Any channel $T_x$ above can be chosen to be of the form
$$
T_x(B)=C^*_x\big(B\otimes I_{\hi'_{n'(x)}}\big)C_x,\qquad B\in\lk,
$$
where $C_x:\,\hi_{n(x)}\to\ki\otimes\hi'_{n'(x)}$ is an isometry and
$\int_\Omega^\oplus\hi'_{n'(x)}\d\mu(x)$ is the direct integral Hilbert space of Theorem \ref{th2} associated to $\M$. Especially, one must have
$\dim\hi_{n(x)}\le\dim \ki\dim\hi'_{n'(x)}$ for $\mu$-almost all $x\in\Omega$.
\end{theorem}

\begin{proof}
Suppose that there exists an 
$\M\in\In(\Sigma,\ki,\hi)$ such that $\M(X,I_\ki)\equiv\Mo(X)$,
and let $Y:\,\hi\to\hd$ and $Y':\,\hi\to\ki\otimes\hi_\oplus'$ be the isometries of Theorems \ref{th1} and \ref{th2}, that is, 
$\M(X,B)\equiv Y'^*(B\otimes\CHII X)Y'$ and
$\Mo(X)\equiv Y^*\CHII X Y\equiv Y'^*(I_\ki\otimes\CHII X)Y'$.
%For any $B\in\lk$ define
%$$
%C(B)\left(\sum_{i=1}^{n}\CHII {X_i}Y\eta_i\right):=
%\sum_{i=1}^{n}\big(B\otimes\CHII {X_i}\big)Y'\eta_i\in\ki\otimes\hi_\oplus'
%$$
%for all $X_i\in\Sigma$ and $\eta_i\in\hi$ where $i=1,2,\ldots,n\in\N_+$.
%Since
%\begin{align*}
%\Big\|C(B)\Big(\sum_i\CHII {X_i}Y\eta_i\Big)\Big\|^2
%&=\Big\<\sum_i(I_\ki\otimes\CHII {X_i})Y'\eta_i\Big|
%\big(B^*B\otimes I_{\hi_\oplus'}\big)
%\sum_j(I_\ki\otimes \CHII {X_j})Y'\eta_j\Big\> \\
%&\le\|B\|^2\Big\<\sum_i(I_\ki\otimes\CHII {X_i})Y'\eta_i\Big|
%\sum_j(I_\ki\otimes \CHII {X_j})Y'\eta_j\Big\> \\
%&=\|B\|^2\Big\|\sum_i\CHII {X_i}Y\eta_i\Big\|^2
%\end{align*}
%it follows that $C(B):\,\hd\to\ki\otimes\hi_\oplus'$ is well-defined bounded linear operator. Moreover, $(I_\ki\otimes\CHII X)C(B)=C(B)\CHII X$ for all $X\in\Sigma$ so that $C(B)$ is decomposable by proposition \ref{prop1},
%$
%C(B)=\int_\Omega^\oplus C_x(B)\d\mu(x).
%$
%Also $B\mapsto C(B)$ is linear, $C(I_\ki)Y=Y'$, and
%$$
%\M(X,B)\equiv Y'^*(B\otimes\CHII X)Y'=Y^*C(I_\ki)^*(B\otimes I_{\hd})(I_\ki\otimes\CHII X)C(I_\ki)Y
%=Y^*T(B)\CHII X Y
%$$
%where $T:\lk\to\li(\hd),\;B\mapsto T(B):=C(I_\ki)^*(B\otimes I_{\hd})C(I_\ki)$
%is obviously a CP channel and, for all $B\in\lk$,
Define
$$
C\left(\sum_{i=1}^{n}\CHII {X_i}Y\eta_i\right):=
\sum_{i=1}^{n}\big(I_\ki\otimes\CHII {X_i}\big)Y'\eta_i\in\ki\otimes\hi_\oplus'
$$
for all $X_i\in\Sigma$ and $\eta_i\in\hi$ where $i=1,2,\ldots,n\in\N_+$.
Since
\begin{align*}
\Big\|C\Big(\sum_i\CHII {X_i}Y\eta_i\Big)\Big\|^2
&=\Big\<\sum_i(I_\ki\otimes\CHII {X_i})Y'\eta_i\Big|
\sum_j(I_\ki\otimes \CHII {X_j})Y'\eta_j\Big\> \\
&=\sum_{i,j}\big\<\eta_i\big|Y'^*(I_\ki\otimes \CHII {X_i\cap X_j})Y'\eta_j\big\> 
=
\sum_{i,j}\big\<\eta_i\big|Y^* \CHII {X_i\cap X_j}Y\eta_j\big\>
\\
&=
\Big\<\sum_i\CHII {X_i}Y\eta_i\Big|
\sum_j\CHII {X_j}Y\eta_j\Big\> 
=
\Big\|\sum_i\CHII {X_i}Y\eta_i\Big\|^2
\end{align*}
it follows that $C:\,\hd\to\ki\otimes\hi_\oplus'$ is well-defined linear isometry, $C^*C=I_{\hd}$. Moreover, $(I_\ki\otimes\CHII X)C=C\CHII X$ for all $X\in\Sigma$ so that $C$ is decomposable by Proposition \ref{prop1},
that is,
$
C=\int_\Omega^\oplus C_x\d\mu(x)
$
where the operators $C_x:\,\hi_{n(x)}\to\ki\otimes\hi'_{n'(x)}$ are isometries for $\mu$-almost all $x\in\Omega$; here we have denoted
$\hi_\oplus'=\int_\Omega^\oplus \hi'_{n'(x)}\d\mu(x)$.
Since $CY=Y'$ and $(I_\ki\otimes\CHII X)C=C\CHII X$ one sees that
$$
\M(X,B)\equiv Y'^*(B\otimes\CHII X)Y'=Y^*C^*(B\otimes I_{\hd})(I_\ki\otimes\CHII X)CY
=Y^*T(B)\CHII X Y
$$
where $T:\lk\to\li(\hd),\;B\mapsto T(B):=C^*(B\otimes I_{\hd})C$
is obviously a CP channel and, for all $B\in\lk$,
$T(B)$ is decomposable, 
$
T(B)=\int_\Omega^\oplus T_x(B)\d\mu(x),
$
where $T_x(B)=C^*_x(B\otimes I_{\hi'_{n'(x)}})C_x$ for $\mu$-almost all $x\in\Omega$.
Let $\Omega_C\subseteq\Omega$ be a $\mu$-measurable set such that $C_x$ is an isometry for all $x\in\Omega_C$ and $\mu(\Omega\setminus\Omega_C)=0$.
Then, for each $x\in\Omega_C$, the map $B\mapsto T_x(B)=C^*_x(B\otimes I_{\hi'_{n'(x)}})C_x$ is clearly a CP channel.
The converse claim is trivial and the proof is complete.
\end{proof}

\begin{example}[Nuclear instruments and EB channels]\rm
\label{nuclearex}
If one chooses $T_x(B)=\tr{\sigma_x B}I_{\hi_{n(x)}}$, $B\in\lk$, 
where $\{\sigma_x\}_{x\in\Omega}\subseteq\mathcal S(\ki)$ is a ($\mu$-measurable) family of states, one gets the {\it nuclear instrument}
$$
\M(X,B)=\int_X\tr{\sigma_x B}\d\Mo(x),\qquad X\in\Sigma,\;B\in\lk,
$$
introduced by Ozawa \cite{Oz85}. 
Its predual instrument is
$$
\M_*(X,\rho)=\int_X\sigma_x\tr{\rho\Mo(\d x)},\qquad X\in\Sigma,\;\rho\in\th
$$
and the associate (predual) channel
$$
\M_*(\Omega,\rho)=\int_\Omega\sigma_x\tr{\rho\Mo(\d x)},\qquad X\in\Sigma,\;\rho\in\th.
$$
is {\it entanglement-breaking (EB)} \cite{HoShRu,HoShWe,Ho08,Teikonkirja}.
Following \cite{HoShWe} one sees that a {\it CP channel is EB if and only if it is the associate channel of a (nonunique) nuclear instrument}. 
It is easy to see that an $\Mo$-compatible instrument is nuclear if and only if its structure vectors can be chosen to be decomposable, that is, $\psi_n^s(x)=\eta^s(x)\otimes\psi_n(x)$ where $\eta^s(x)\in\ki$, $\sum_{s=1}^{\dim\ki}\|\eta^s(x)\|^2=1$, and vectors $\psi_n(x)$ are the structure vectors of $\Mo$.

For example, the instrument of Davies and Lewis \cite[Theorem 1]{DaLe}  is a special case of nuclear instruments so that, for any POVM $\Mo$, there exists a (nontrivial) nuclear instrument implementing $\Mo$.
Later we show that any $\Mo$-compatible instrument is nuclear if $\Mo$ is of rank 1 (thus generalizing Corollary 1 of \cite{HeWo}).

%Ozawan \cite{Oz85} paperin
%Theorem 4.4: For any POVM $\Mo\in\O(\Sigma,\hi)$ and $\Mo$-integrable\footnote{That is, if for all $B\in\lh$, $x\mapsto\tr{\rho_xB}$ is $\d\tr{\rho\Mo(x)}$-integrable for all $\rho\in\sh$ and $\int_X\tr{\rho_xB}\d\tr{\rho\Mo(x)}\equiv\tr{\rho_X B}$ for some $\rho_X\in\th$.}
% family $\{\rho_x\,|\,x\in\Omega\}\subseteq\sh$ there is a CP instrument (the so-called {\it nuclear instrument}) $\M\in\In(\Sigma,\hi,\hi)$ with $\Mo_{I_\hi}=\Mo$ such that 
%$$
%\tr{\rho\M(X,B)}=\tr{\M_*(X,\rho)B}=\int_X\tr{\rho_x B}\tr{\rho\Mo(\d x)}
%$$
%for all $\rho\in\th$ and $B\in\lh$.

\end{example}

Let $\Mo\in\O(\Sigma,\hi)$, $\M\in\In(\Sigma,\ki,\hi)$, and assume that
$\M(X,I_\ki)\equiv\Mo(X)$. Let $\hd=\int_\Omega^\oplus\hi_{n(x)}\d\mu(x)\subseteq L^2(\mu,\li)$ and 
$\hi_\oplus'=\int_\Omega^\oplus\hi_{n'(x)}\d\mu(x)\subseteq L^2(\mu,\li)$ be the direct integral Hilbert spaces of Theorems \ref{th1} and \ref{th2}, and
recall the notations and definitions of Subsection \ref{POVM} and Remark \ref{remu} for $\Mo$ and $\M$, respectively:
\begin{align*}
\Mo(X)&\supseteq\int_X\A^1(x)^*\A^1(x)\d\mu(x)=
\int_X\sum_{k=1}^{n(x)}\A_k(x)^*\A_k(x)\d\mu(x),\\
\A^1(x)&:=\sum_{k=1}^{n(x)}\kb{b_k}{d_k(x)}=
\sum_{n=1}^{\dim\hi}\kb{\psi_n(x)}{h_n},\qquad\A_k(x):=\kb1{d_k(x)},\\
\M(X,I_\ki)&\supseteq%\sum_{s=1}^{\dim\ki}\int_X\A^s(x)^*\A^s(x)\d\mu(x)=
\int_X \sum_{k=1}^{n'(x)} \A'_k(x)^*B\A'_k(x)\d\mu(x),\\
\A'_k(x)%=\sum_s\kb{k_s}{b_k}\A^s(x)
&:=\sum_{s=1}^{\dim\ki}\sum_{n=1}^{\dim\hi}\<b_k|\psi_n^s(x)\>\kb{k_s}{h_n}=\sum_{s=1}^{\dim\ki}\kb{k_s}{d^s_k(x)}.
\end{align*}
Let $T$ be the decomposable CP channel of Theorem \ref{compa}.
Immediately we see that 
$\<\psi^s_n(x)|\psi^t_m(x)\>=\<\psi_n(x)|T_x(\kb s t)\psi_m(x)\>$
and \eqref{solve} holds. Moreover, one
obtains a {\it pointwise Kraus decomposition for} $T$:
\begin{align*}
\M(X,I_\ki)&\supseteq
\int_X \sum_{k=1}^{n'(x)} \A'_k(x)^*B\A'_k(x)\d\mu(x)
=\int_X \A^1(x)^*\left[\sum_{k=1}^{n'(x)} \A^T_k(x)^*B\A^T_k(x)\right]\A^1(x)\d\mu(x)
\\
T_x(B)&
=C^*_x(B\otimes I_{\hi_{n'(x)}})C_x=\sum_{k=1}^{n'(x)}
C^*_x(B\otimes{\kb{b_k}{b_k}})C_x=
\sum_{k=1}^{n'(x)} \A^T_k(x)^*B\A^T_k(x),
\end{align*}
where (when $C_x$ is an isometry) 
$$
\A^T_k(x):=\sum_{s=1}^{\dim\ki}\kb{k_s}{k_s\otimes b_k}C_x
=\sum_{s=1}^{\dim\ki}\kb{k_s}{D^s_k(x)}
$$
is a bounded operator from $\hi_{n(x)}$ to $\ki$ and the generalized vectors $D^s_k(x):=C_x^*(k_s\otimes b_k)$ of the channel $T_x$
are now elements of the Hilbert space $\hi_{n(x)}$.
Note also that then $\sum_k \A^T_k(x)^*\A^T_k(x)=I_{\hi_{n(x)}}$.
Since $C_x\psi_n(x)=C_x(Yh_n)(x)=(Y'h_n)(x)=\sum_s k_s\otimes\psi_n^s(x)$ we get the following relations:
\begin{align*}
\A'_k(x)&=\A^T_k(x)\A^1(x),\qquad
\A^T_k(x)\psi_n(x)=\sum_{s=1}^{\dim\ki}\<b_k|\psi^s_n(x)\>k_s,\\
d_k^s(x)&=\sum_{l=1}^{n(x)}\<\A_k^T(x)b_l|k_s\>d_l(x)
=\sum_{l=1}^{n(x)}\<C_x b_l|k_s\otimes b_k\>d_l(x)
=\sum_{l=1}^{n(x)}\<b_l|D^s_k(x)\>d_l(x).
\end{align*}
The last equation gives us a relation between the generalized vectors $d_k^s(x)$, $D^s_k(x)$, and $d_k(x)$ of $\M$, $T$, and $\Mo$, respectively.
In addition, it provides 
an {\it effective tool for determining $\Mo$-compatible instruments:} If the generalized vectors $d_l(x)$ of $\Mo$ are given, take a (measurable collection of) complex `matrices' $(c^x_{l,sk})$
such that the orthogonality condition
$$
\sum_{s,k}c^x_{l,sk}\overline{c^x_{l',sk}}=\delta_{ll'}.
$$
is satisfied.
Then the instrument $\M\in\In(\Sigma,\ki,\hi)$ defined by the generalized vectors
$$
d_k^s(x):=\sum_{l=1}^{n(x)}c^x_{l,sk} d_l(x)
$$
is $\Mo$-compatible.
Especially, if $n(x)=1$ one sees that
$d_k^s(x)=c^x_{1,sk}d_1(x)$ where $c^x_{1,sk}\in\C$ and
$\sum_{s,k}|c^x_{1,sk}|^2=1$.

\begin{corollary}\label{corolla1}
Let $\Mo\in\O(\Sigma,\hi)$ and $\M\in\In(\Sigma,\ki,\hi)$. If
$\M(X,I_\ki)\equiv\Mo(X)$ then
there exist CP channels $\Phi^X:\,\lk\to\lh$, $X\in\Sigma$, such that
$$
\M(X,B)=\sqrt{\Mo(X)}\,\Phi^X(B)\sqrt{\Mo(X)},\qquad X\in\Sigma,\;\;B\in\lk.
$$
\end{corollary}

\begin{proof}
Assume that $\M(X,I_\ki)\equiv\Mo(X)$ and let $T$ be the channel of Theorem \ref{compa} so that $\M(X,B)\equiv Y^*T(B)\CHII X Y$.
Fix $X\in\Sigma$ and define, for all  $\eta\in\hi$,
$$
E_X\Big(\sqrt{\Mo(X)}\,\eta\Big):=\CHII{X}Y\eta
$$
for which 
\begin{align*}
\Big\|E_X\Big(\sqrt{\Mo(X)}\,\eta\Big)\Big\|^2&=\big\<\CHII{X}Y\eta\big|\CHII{X}Y\eta\big\>
=\big\<\eta\big|Y^*\CHII{X}Y\eta\big\>
=\<\eta\big|\Mo(X)\eta\big\> =
\Big\|\sqrt{\Mo(X)}\,\eta\Big\|^2,
\end{align*}
so that $E_X$ is an isometry from $\sqrt{\Mo(X)}\,\hi$ to $\hd$.
Extend $E_X$ to some isometry $\ov E_X:\,\hi\to\hd$, ${\ov E}_X^*\ov E_X=I_\hi$ (which clearly exists since $Y:\,\hi\to\hd$ is isometric).
Define a channel $\Phi^X:\,\lk\to\lh,\;B\mapsto\Phi^X(B):=\ov E_X^* T(B)\ov E_X$ for which
$$
\Big\<\eta\Big|\sqrt{\Mo(X)}\,\Phi^X(B)\sqrt{\Mo(X)}\,\eta\Big\>
=\big\<\CHII{X}Y\eta\big|T(B)\CHII{X}Y\eta\big\>
=\big\<\eta\big|\M(X,B)\eta\big\>
$$
for all $\eta\in\hi$ and $B\in\lk$. Hence, the corollary follows.
\end{proof}

\begin{remark}\rm
In the context of Theorem \ref{compa} and Corollary \ref{corolla1}, we note the following  facts: Let $\M\in\In(\Sigma,\ki,\hi)$ be an {\it arbitrary} instrument and $\Mo(X)=\M(X,I_\ki)$ its associate observable with the minimal Naimark isometry $Y:\hi\to\hd$ of Theorem \ref{th1}.
\begin{enumerate}
\item There exists a CP channel $T$ such that $\M(X,B)\equiv Y^*
\ov\M(X,B)Y$ where 
$$
\ov\M:\,\Sigma\times\lk\to\li(\hd),\;(X,B)\mapsto
\ov\M(X,B):=T(B)\CHII X
$$
is an instrument with the canonical spectral measure as its associate observable and $T$ as its associate channel.
\item If $\Mo$ is projection valued then $YY^*=I_{\hd}$, $Y\hi=\hd$, and 
\begin{align*}
\M(X,B)&\equiv Y^*T(B)YY^*\CHII XY=T_Y(B)\Mo(X)=
\Mo(X)T_Y(B)\Mo(X) \\
&=\sqrt{\Mo(X)}\,T_Y(B)\sqrt{\Mo(X)}
\end{align*}
where $T_Y:=Y^*T Y:\,\lk\to\lh$ is a CP channel commuting with $\Mo$.
The channels $\Phi^X$ can now be chosen to be the same $T_Y$ for all $X\in\Sigma$.
\end{enumerate}
\end{remark}

\begin{example}[Predual instruments] \rm
Let $\Mo$, $Y$, $\M$, $T$, and $\Phi^X$ be as in Theorem \ref{compa}
and Corollary \ref{corolla1}.
Let
$\M_*:\,\Sigma\times\th\to\tk$, $T_*:\,\mathcal T(\hd)\to\tk$, and $\Phi^X_*:\,\th\to\tk$ be the predual mappings of $\M$, $T$, and $\Phi^X$, respectively.
Now, for all $X\in\Sigma$, $\rho\in\th$, and $B\in\lk$, 
\begin{align*}
\tr{\M_*(X,\rho)B}=\tr{\rho\M(X,B)}=
\tr{\rho Y^*\CHII XT(B)\CHII X Y}=
\tr{\CHII X Y\rho Y^*\CHII XT(B)}%=\tr{T_*(\CHII X Y\rho Y^*)B}
\end{align*}
so that 
$$
\M_*(X,\rho)=T_*(\CHII X Y\rho Y^*\CHII X),
$$
that is, the L\"uders operation $\rho\mapsto\CHII X\tilde\rho\CHII X$ of the larger system (with the Hilbert space $\hd$) first operates to a subsystem state $\tilde\rho=Y\rho Y^*\in\mathcal S(Y\hi)$, $\hi \cong Y\hi\subseteq\hd$ and, then, the resulting (nonnormalized) state is transformed by the channel $T_*$.

Similarly, we get
$$
\M_*(X,\rho)=\Phi^X_*\big(\sqrt{\Mo(X)}\,\rho\sqrt{\Mo(X)}\big)
$$
which can be interpreted as a combination of a L\"uders operation
$\rho\mapsto\sqrt{\Mo(X)}\,\rho\sqrt{\Mo(X)}$ and a channel $\Phi^X_*$.
In the both above interpretations of $\M_*$, one can see the action of the channel as adding quantum noise to the L\"uders operation.

\end{example}

%$$
%\tr{\CHII X Y\rho Y^*}=
%\tr{\Mo(X)\rho}=
%\tr{\M_*(X,\rho)}=\tr{T_*(\CHII X Y\rho Y^*)}
%$$

%\begin{align*}
%\M(X,B)
%&=\sum_{n,m=1}^{\dim\hi}\sum_{s,t=1}^{\dim\ki}B_{st}\int_{X}\<\psi^s_n(x)|\psi^t_m(x)\>\d\mu(x)\kb{h_n}{h_m} \\
%&=\sum_{n,m=1}^{\dim\hi}\sum_{s,t=1}^{\dim\ki}B_{st}\int_{X}
%\big\<k_s\otimes\psi_n(x)\big|C(x)^*C(x)\big(k_t\otimes\psi_m(x)\big)\big\>\d\mu(x)\kb{h_n}{h_m} 
%\end{align*}
%

%Moreover, by denoting $E=C^*C$,
%$$
%\delta_{nm}=\<\psi_n|\psi_m\>
%=\<\psi_n|({\rm tr}_\k E)\psi_m\>=
%\sum_{s=1}^{\dim\ki}
%\big\<k_s\otimes\psi_n\big|E\big(k_s\otimes\psi_m\big)\big\>
%$$
%and THE STRUCTURE of $\M$ is
%$$
%\<\psi^s_n|\CHII{X\cap Y}\psi^t_m\>
%=
%\<\CHI X\psi^s_n|\CHI Y\psi^t_m\>
%=
%\big\<k_s\otimes\CHI X \psi_n\big|E\big(k_t\otimes\CHI Y\psi_m\big)\big\>
%$$
%implying
%\begin{align*}
%\<\CHI X\psi_n|({\rm tr}_\k E)\CHI Y\psi_m\>&=
%\sum_s\big\<k_s\otimes\CHI X \psi_n\big|E\big(k_s\otimes\CHI Y\psi_m\big)\big\>
%=\sum_s\<\psi^s_n|\CHII{X\cap Y}\psi^s_m\>
%\\
%&=
%\sum_s\int_{X\cap Y}
%\big\<k_s\otimes\psi_n(x)\big|C(x)^*C(x)\big(k_s\otimes\psi_m(x)\big)\big\>\d\mu(x)
%\\
%&=\int_{X\cap Y}\<\psi_n(x)|\psi_m(x)\>\d\mu(x)=
%\<\CHI X\psi_n|\CHI Y\psi_m\>
%\end{align*}
%so that one must have
%$$
%\boxed{
%{\rm tr}_\k E=I_{\hd}\quad\Longleftrightarrow\quad
%{\rm tr}_\k E(x)=I_{n(x)}
%}
%$$

\section{Measurement theory}

Following e.g.\ \cite[Definition 3.1]{Oz84} or \cite{BuLaMi}, we define
{\it a measuring process (a measurement model \emph{or} a premeasurement}) ${\bf M}$ of a POVM $\Mo\in\O(\Sigma,\hi)$ as a 4-tuple ${\bf M}=\<\hi',\Po,\sigma,U\>$ consisting of a (possibly nonseparable) Hilbert space $\hi'$ attached to the probe, a PVM $\Po:\,\Sigma\to\li(\hi')$ (the pointer observable), an initial state $\sigma\in\mathcal S(\hi')$ of the probe, and a unitary operator $U\in\mathcal L(\hi\otimes\hi')$ (the measurement interaction) satisfying the relation
$$
\tr{\rho\Mo(X)}=\tr{U(\rho\otimes\sigma)U^*\big(I_\hi\otimes\Po(X)\big)},\hspace{0.5cm}\rho\in\th,\;X\in\Sigma,
$$
or, equivalently,
$$
\Mo(X)=E_\sigma\big[U^*\big(I_\hi\otimes\Po(X)\big)U\big],\hspace{0.5cm}X\in\Sigma,
$$
where $E_\sigma:\,\mathcal L(\hi\otimes\ki)\to\lh$ is a (normal) CP map defined by the formula
$$\tr{\rho E_\sigma(A)}:=\tr{(\rho\otimes\sigma)A},\hspace{0.5cm}\rho\in\th,\; A\in\mathcal L(\hi\otimes\ki).
$$
%Especially, $E_\sigma(B\otimes I_\ki)=B$ for all $B\in\lh$ \cite[Lemma 1.1]{Davies&LUUISH KAI?}.
Moreover, ${\bf M}=\<\hi',\Po,\sigma,U\>$ is said to be {\it pure} (or {\it normal}) if $\sigma=\kb\xi\xi$ for some $\xi\in\hi'$, $\|\xi\|=1$. 

A measuring process ${\bf M}=\<\hi',\Po,\sigma,U\>$ of $\Mo$ defines an $\Mo$-compatible instrument $\M\in \In(\Sigma,\,\hi,\,\hi)$
by
\begin{equation}
\label{mittaus}
\M(X,B):=E_\sigma\big[U^*\big(B\otimes\Po(X)\big)U\big],\hspace{0.5cm}X\in\Sigma,\;B\in\lh.
\end{equation}
Now its predual instrument
$$
\M_*(X,\rho)={\rm tr}_{\hi'}\big[U(\rho\otimes\sigma)U^*\big(I_\hi\otimes\Po(X)\big)\big],\hspace{0.5cm}X\in\Sigma,\;\rho\in\th,
$$
%is an {\it ancilla form} of $\M_*$
so that $\tr{\rho\Mo(X)}=\tr{\M_*(X,\rho)}$.
Two measuring processes are {\it statistically equivalent} if they define the same instrument. 
%Moreover, two (statistically equivalent) measuring processes ${\bf M}=\<\hi',\Po,\sigma,U\>$ and
%$\ov{\bf M}=\<\ov\hi',\ov\Po,\ov\sigma,\ov U\>$ are {\it unitarily equivalent} if there exists a unitary operator $W:\,\ov\hi'\to\hi'$
%such that $\ov U=(I_\hi\otimes W)^*U(I_\hi\otimes W)$, $\ov\sigma=W^*\sigma W$, and 
%$\ov\Po(X)=W^*\Po(X)W$ for all $X\in\Sigma$.

A fundamental result of Ozawa \cite[Theorem 5.1]{Oz84}
is that for each instrument $\M\in\In(\Sigma,\hi,\hi)$ there exists a {\it pure} measuring process ${\bf M}=\<\hi',\Po,\kb\xi\xi,U\>$ of the associate observable $\Mo$ of $\M$, that is, $\M$ is of the form \eqref{mittaus}. Then we say that $\bf M$ is a {\it pure realization of the instrument $\M$}. 
Moreover, we say that a pure realization ${\bf M}=\<\hi',\Po,\kb\xi\xi,U\>$ of $\M$ is {\it minimal} if for each pure realization
$\ov{\bf M}=\<\ov\hi',\ov\Po,\kb{\ov\xi}{\ov\xi},\ov U\>$ of $\M$ there exists an isometry $\hi'\to\ov\hi'$. Physically this means that the ancillary space $\hi'$ of a minimal $\bf M$ is (up to a unitary equivalence) the smallest possible Hilbert space, i.e.\ there are no unnecessary degrees of freedom in the measuring process.
Next we consider minimal pure realizations of instruments.

\begin{example}\rm
Let $\M\in\In(\Sigma,\hi,\hi)$ with the isometry $Y:\,\hi\to\hi\otimes\hd$ of Theorem \ref{th2}.
Let $\e$ be an ON basis of $\hi$ and fix a unit vector 
$\xi_\oplus\in\hd$.
Since $\{Yh_n\}_{n=1}^{\dim\hi}$ and $\{h_n\otimes{\xi_\oplus} \}_{n=1}^{\dim\hi}$  are ON sets (with the same cardinality) one can define a unitary operator $U_{\xi_\oplus} \in\li(\hi\otimes\hd)$ by setting
$U_{\xi_\oplus} (h_n\otimes{\xi_\oplus} ):=Yh_n$. %=:\psi_n=:\sum_{k=1}^{\dim\hi}k_s\otimes\psi_n^s$. 
Obviously, $U_{\xi_\oplus} (\psi\otimes{\xi_\oplus} )=Y\psi$ for all $\psi\in\hi$ and we get, for all $\fii,\,\psi\in\hi$, $X\in\Sigma$, and $B\in\lh$,
\begin{eqnarray*}
\<\fii|E_{\kb{\xi_\oplus} {\xi_\oplus} }\big[U_{\xi_\oplus} ^*\big(B\otimes\CHII X\big)U_{\xi_\oplus} \big]\psi\>
&=&\tr{U_{\xi_\oplus} (\kb{\psi\otimes{\xi_\oplus} }{\fii\otimes{\xi_\oplus} })U_{\xi_\oplus} ^*\big(B\otimes\CHII X\big)} \\
&=&\<\fii|Y^*\big(B\otimes\CHII X\big)Y\psi\>
=\<\fii|\M(X,B)\psi\>
\end{eqnarray*}
so that ${\bf M}^\M_{U_{\xi_\oplus} }:=\<\hd,[X\mapsto\CHII X],\kb{\xi_\oplus} {\xi_\oplus} ,U_{\xi_\oplus} \>$ is a minimal pure realization of $\M$ by the next Theorem \ref{purereali}.
\end{example}

Let $\M\in\In(\Sigma,\hi,\hi)$ with the (minimal) isometry $Y:\,\hi\to\hi\otimes\hd$ of Theorem \ref{th2}, i.e.\ $\M(X,B)=Y^*(B\otimes\CHII X)Y$ and the linear space $${\mathcal V}:={\rm lin}\big\{(B\otimes\CHII X)Y\psi\,\big|\,B\in\lh,\,X\in\Sigma,\,\psi\in\hi\big\}$$ is dense in $\hi\otimes\hd$.
\begin{theorem}\label{purereali}
For any pure realization $\<\hi',\Po,\kb{\xi}{\xi},U\>$ of $\M$ there exists a unique isometry $W:\,\hd\to\hi'$ such that
$\Po(X)W=W\CHII X$ for all $X\in\Sigma$ and
$(I_{\hi}\otimes W)Y\psi=U(\psi\otimes\xi )$ for all $\psi\in\hi$.
The ancillary space $\hi'$ is unitarily equivalent with the dilation space $\hd$ if and only if $W$ is unitary.
Finally, $\dim\hi'\ge\dim\hd\ge\dim\hi$ and, if $\hd$ is not separable, then $\hi'$ cannot be separable.
\end{theorem}

\begin{proof}
Let  $\<\hi',\Po,\kb{\xi}{\xi},U\>$ be a pure realization of $\M$.
Since, for all $\fii,\,\psi\in\hi$, $X\in\Sigma$, and $B\in\lh$,
\begin{eqnarray*}
\<\fii|\M(X,B)\psi\>
%\tr{(\kb\psi\fii\otimes\sigma)\big[U^*\big(B\otimes\Po(X)\big)U\big]}
=
\big\<U(\fii\otimes \xi )\big|\big(B\otimes\Po(X)\big)U(\psi\otimes\xi )\big\> 
=\<Y\fii|\big(B\otimes\CHII X\big)Y\psi\>
\end{eqnarray*}
one can define a linear map
$$
\ov W\Big(\big(B\otimes\CHII X\big)Y\psi\Big):=\big(B\otimes\Po(X)\big)U(\psi\otimes\xi )
$$
on $\mathcal V$ which is well defined and extends to an isometry $\ov W:\,\hi\otimes\hd\to\hi\otimes\hi'$ by the usual calculation:
\begin{eqnarray*}
\Big\|\ov W\Big(\sum_i\big(B_i\otimes\CHII {X_i}\big)Y\psi_i\Big)\Big\|^2
&=&\sum_{i,j}\big\<U(\psi_i\otimes\xi )\big|\big(
B_i^*B_j\otimes\Po(X_i\cap X_j)\big)U(\psi_j\otimes\xi )\big\>\\
&=&\sum_{i,j}\<Y\psi_i|\big(B_i^*B_j\otimes\CHII{X_i\cap X_j}\big)Y\psi_j\> =
\Big\|\sum_i\big(B_i\otimes\CHII {X_i}\big)Y\psi_i\Big\|^2.
\end{eqnarray*}
Since, %$\ov W^*\ov W=I_{\hi\otimes\hd}$ and
for all $B\in\lh$, $X\in\Sigma$, and $\big(B'\otimes\CHII{X'}\big)Y\psi\in\mathcal V$,
\begin{eqnarray*}
&&
\big(B\otimes\Po(X)\big)\ov W\Big(\big(B'\otimes\CHII{X'}\big)Y\psi\Big)=
\big(B\otimes\Po(X)\big)\big(B'\otimes\Po(X')\big)U(\psi\otimes\xi )\\
&&\qquad=\big(BB'\otimes\Po(X\cap X')\big)U(\psi\otimes\xi )
=\ov W\Big(\big(B\otimes\CHII X\big)\big(B'\otimes\CHII{X'}\big)Y\psi\Big)
\end{eqnarray*}
one gets 
$
\big(B\otimes\Po(X)\big)\ov W=\ov W\big(B\otimes\CHII X\big)
$
which implies (put $X=\Omega$) that
$
\big(B\otimes I_{\hi'}\big)\ov W=\ov W\big(B\otimes I_\hd\big)
$
and, hence, $\ov W=I_{\hi}\otimes W$ where $W:\,\hd\to\hi'$ is an isometry for which $\Po(X)W=W\CHII X$ for all $X\in\Sigma$.
Obviously, $(I_{\hi}\otimes W)Y\psi=U(\psi\otimes\xi )$ follows from the definition of $\ov W$.

Suppose then that
$W':\,\hd\to\hi'$ is such that
$\Po(X)W'=W'\CHII X$ for all $X\in\Sigma$ and
$(I_{\hi}\otimes W')Y\psi=U(\psi\otimes\xi )$ for all $\psi\in\hi$.
Then
$$
(I_\hi\otimes W')\Big(\big(B\otimes\CHII X\big)Y\psi\Big)=\big(B\otimes\Po(X)\big)U(\psi\otimes\xi )
=(I_\hi\otimes W)\Big(\big(B\otimes\CHII X\big)Y\psi\Big)
$$
so that $I_\hi\otimes W'=I_\hi\otimes W$ by the density of $\mathcal V$, and thus $W'=W$ showing the uniqueness of $W$.

Since $W^*W=I_{\hd}$, it follows that $R=WW^*$ is a projection on $\hi'$, and
$W\hd=R\hi'$. Thus, $R\hi'$ is unitarily equivalent with $\hd$, and
if also $\hi'$ is unitarily equivalent with $\hd$ it follows that 
$R\hi'$ is unitarily equivalent with $\hi'$ implying $R=I_{\hi'}$ and 
$W$ is unitary.
The remaining claims are trivial.
\end{proof}
Immediately we get the following corollary:

\begin{corollary}\label{purecorolla}
Let ${\bf M}=\<\hi',\Po,\kb{\xi}{\xi},U\>$ be a pure realization of $\M$.
Then ${\bf M}$ is minimal if and only if there exists a unitary operator $W:\,\hd\to\hi'$ such that 
\begin{itemize}
\item $W^*\Po(X)W=\CHII X$ for all $X\in\Sigma$,
\item $(I_\hi\otimes W)^*U(I_\hi\otimes W)=U_{\xi_\oplus}$ where $\xi_\oplus=W^*\xi$ and $U_{\xi_\oplus}\in\li(\hi\otimes\hd)$ is a unitary operator for which $U_{\xi_\oplus} (\psi\otimes{\xi_\oplus} )=Y\psi$ for all $\psi\in\hi$.
\end{itemize}
Hence, a minimal pure realization of $\M$ is always {\em unitarily equivalent} (and can be identified) with some 
${\bf M}^\M_{U_{\xi_\oplus} }=\<\hd,[X\mapsto\CHII X],\kb{\xi_\oplus} {\xi_\oplus} ,U_{\xi_\oplus} \>$.
\end{corollary}

%\newpage

\begin{remark}\rm\label{rem7}
In this remark, we determine the structure of pure realizations which are not necessarily minimal.

Let ${\bf M}=\<\hi',\Po,\kb{\xi}{\xi},U\>$ be a pure realization of $\M$ with the isometry $W$ of Theorem \ref{purereali}. Let $R=WW^*$ be the projection $R:\,\hi'\to R\hi'\cong\hd$, $R^\perp:=I_{\hi'}-R$, and
$Y':=(I_\hi\otimes W)Y$.
It is easy to verify that, for all $X\in\Sigma$ and $B\in\lh$,
\begin{enumerate}
\item
$W^*\Po(X)W=\CHII X$,
\item 
$\Po(X)R=R\Po(X)=R\Po(X)R=W\CHII XW^*$,
\item
$\Po(X)=R\Po(X)R+R^\perp\Po(X)R^\perp$,
\item
$E_{\kb{\xi} {\xi} }\big[U ^*\big(B\otimes\Po(X)\big)U \big]
=Y'^*\big(B\otimes\Po(X)\big)Y'
=Y^*\big(B\otimes\CHII X\big)Y
=\M(X,B)$,
\item $(I_\hi\otimes R)Y'=Y'$,
\item $\M(X,B)=Y'^*\big(B\otimes R\Po(X)R\big)Y'$,
\item $(I_\hi\otimes R)U(\psi\otimes\xi)=U(\psi\otimes\xi)$, $\psi\in\hi$.
\end{enumerate}
We have two cases a) $\xi\in R\hi'$ and b) $\xi\notin R\hi'$.
In the a) case one sees that, if the compound system is in a (pure) state $\psi\otimes\xi\in\hi\otimes R\hi'$ before the measurement, then the measurement coupling $U$ transforms it to the state $U(\psi\otimes\xi)\in\hi\otimes R\hi'$ by item (7) above. 
Now, instead of $\Po$, one can equally well use the projected pointer observable (PVM) $\Sigma\ni X\mapsto R\Po(X) R\in\li(R\hi')$ of the smaller probe system, see (6); the PVM $X\mapsto R^\perp\Po(X)R^\perp$ is irrelevant in the measurement model $\bf M$, see (3).

In the case b), one can slightly modify $U$ in the following way:
Let $\tilde\xi\in R\hi$ be a unit vector and
$W_0\in\li(\hi')$ a unitary rotation operator which maps $\xi$ to $\tilde \xi$ 
$(=W_0\xi)$.
Define a unitary operator $\tilde U:=U(I_\hi\otimes W_0^*)$ for which $\tilde U(\psi\otimes\tilde \xi)=U(\psi\otimes\xi)$ for all $\psi\in\hi$.
From (7) one sees that 
$(I_\hi\otimes R)\tilde U(\psi\otimes\tilde \xi)=\tilde U(\psi\otimes\tilde\xi)$, $\psi\in\hi$.
Hence, we get a measuring process $\<\hi',\Po,\kb{\tilde\xi}{\tilde\xi},\tilde U\>$ of type a), $\tilde\xi\in R\hi'$, which is statistically equivalent with $\bf M$.
\end{remark}

\begin{example}[The standard model]\rm \label{stmo}
In some physically relevant cases, the measurement coupling $U$ is given and then one determines observables which can be measured by using $U$ and varying pointer observables and initial probe states. For example, in the (generalized) standard model of quantum measurement theory  $U$ is of the form $U=e^{i\lambda(A\otimes B)}$ (i.e.\ $U$ is the standard measurement coupling)
where $A$ is a selfadjoint operator of the system on which a measurement is perfomed, $B$ is a selfadjoint operator of the ancillary space, and $\lambda\in\R$ is a coupling constant (see, e.g.\ \cite{BuLa95,HaLaSc}). It can be shown \cite{BuLa95} that, when the pointer observable is a selfadjoint operator then one usually measures a smeared (or unsharp) version of $A$. An important example is a standard measurement of a (fuzzy) position observable.
In this example, we generalize the standard measurement model for arbitrary POVMs.

From Remark \ref{rem7} and Corollary \ref{purecorolla}, one sees that we can study a pure measurement model ${\bf M}=\<\hd,[X\mapsto \CHII X],\kb{\xi}{\xi},U\>$ where $\hd$ is an arbitrary direct integral Hilbert space and
$U$ is the standard measurement coupling, $U=e^{i\lambda(A\otimes B)}$.
Now $A$ and $B$ are selfadjoint operators of $\hi$ and $\hd$, respectively. 
Define 
$Y\psi:=U(\psi\otimes\xi)$ for all $\psi\in\hi$.
If $\Mo_A\in\O\big(\bo\R,\hi\big)$ is a spectral measure of $A$ we get, for all $\psi\in\hi$,
$$
Y\psi=\int_\R\Mo_A(\d a)\psi\otimes\xi(a)
$$
where $\xi(a):=e^{ia\lambda B}\xi$ is a unit vector for each $a\in\R$.
Now the realized instrument 
$\M(X,B)=Y^*(B\otimes\CHII X)Y$, $X\in\Sigma$, $B\in\lh$, or
\begin{eqnarray*}
\<\psi|\M(X,B)\psi\>&=&\<Y\psi|(B\otimes\CHII X)Y\psi\>
=
\int_\R\int_\R
\<\Mo_A(\d a)\psi\otimes\xi(a)|(B\otimes\CHII X)\Mo_A(\d a')\psi\otimes\xi(a')\>\\
&=&
\int_\R\int_\R
\<\xi(a)|\CHII X\xi(a')\>
\<\psi|\Mo_A(\d a)B\Mo_A(\d a')\psi\>.
\end{eqnarray*}
Especially, the measured observable $\Mo\in\O(\Sigma,\hi)$ is given by
$$
\<\psi|\Mo(X)\psi\>=
\<\psi|\M(X,I_\hi)\psi\>=
\int_\R\|\CHII X\xi(a)\|^2
\<\psi|\Mo_A(\d a)\psi\>.
$$
that is,
$$
\Mo(X)=\int_\R K(X,a)\Mo_A(\d a)
$$
where $K:\,\Sigma\times\R\to[0,1]$,
$$
K(X,a):=\|\CHII X\xi(a)\|^2=\big\|\CHII Xe^{ia\lambda B}\xi\big\|^2
=\int_X k(x,a) \d\mu(x),
$$
is a {\it Markov kernel} (or a conditional or transition probability)
{\it with the density} $k(x,a):=\big\|\big(e^{ia\lambda B}\xi\big)(x)\big\|^2$:
for all $a\in\R$, 
$$
\Sigma\ni X\mapsto K(X,a)\in[0,1]
$$
is a probability measure and, for all $X\in\Sigma$,
$$
\R\ni a\mapsto K(X,a)\in[0,1]
$$
is continuous.
If $B$ (its spectral measure) commutes with the canonical spectral measure $X\mapsto\CHII X$ then $K(X,a)=\|\CHII X\xi\|^2$ does not depend on $a$ and the measured POVM $\Mo(X)=\|\CHII X\xi\|^2I_\hi$ is trivial. Hence, to measure some nontrivial POVM, 
the probe observables $B$ and $X\mapsto\CHII X$ must be `complementary' in some sense (e.g.\ position and momentum \cite{BuLa95}).

Finally we note that, on the first hand,
since (as a PVM) $\Mo_A$ is commutative, $\Mo$ is also commutative. On the other hand, some important POVMs (e.g.\ the canonical phase observable) are (even totally) noncommutative.
We may conclude that the standard measurement couplings cannot be used in any measurement of  noncommutative observables.
\end{example}

\subsection{Posterior states}\label{postpost}
Let $\M\in\In(\Sigma,\ki,\hi)$ with the (minimal) point- and setwise Kraus operators of Theorem \ref{th2}, i.e., we may write either
\begin{eqnarray*}
\<\fii|\M(X,B)\psi\>
= \int_X \sum_{k=1}^{n(x)} \<\A_k(x)\fii|B\A_k(x)\psi\>\d\mu(x),
\hspace{0.5cm}\fii,\,\psi\in V_\e,
\end{eqnarray*}
or
$$
\M(X,B)=\sum_{k=1}^{n(X)} \A_k(X)^*B\A_k(X)
\qquad\text{(ultraweakly) }
$$ 
where,
for all $x\in\Omega$, the operators $\A_k(x):\,V_\e\to\ki$ are linearly independent, and
for all $X\in\Sigma$, the bounded operators $\A_k(X):\,\hi\to\ki$ are linearly independent.
Note that, usually in quantum measurement theory, one assumes $\ki=\hi$.

Let $\rho\in\sh$ be a (fixed) `initial state' of the system. 
Then the `measurement' probability measure is
$$
X\mapsto\mu^\M_\rho(X):=\tr{\rho\Mo_{I_\ki}(X)}=\tr{\rho\M(X,I_\ki)}=\tr{\M_*(X,\rho)}.
$$
Denote by $w_\rho$ the (nonnegative) Radon-Nikod\'ym derivative of $\mu^\M_\rho$ with respect to $\mu$, i.e.\ $\d\mu^\M_\rho(x)=w_\rho(x)\d\mu(x)$.
Following \cite{Oz85}, we define
a {\it conditional expectation 
$\{E(B|x)\,|\,x\in\Omega\}$ of $B\in\lk$ with respect to} $(\M,\rho)$ as a (complex-valued) measurable function 
$x\mapsto E(B|x)$ such that, for all $X\in\Sigma$,
$$
\int_X E(B|x)\d\mu^\M_\rho(x)=\tr{\rho\M(X,B)}=\tr{B\M_*(X,\rho)},
$$
and we call $\{\rho_x\,|\,x\in\Omega\}\subseteq\mathcal S(\ki)$ a family of {\it posterior states with respect to $(\M,\rho)$} if, for all $B\in\lh$, one has $\tr{\rho_x B}=E(B|x)$ for $\mu^\M_\rho$-almost all $x\in\Omega$.
If $\{\rho_x'\,|\,x\in\Omega\}\subseteq\mathcal S(\ki)$ is another family of posterior states then $\rho_x'=\rho_x$ for $\mu^\M_\rho$-almost all $x\in\Omega$.

Let $\e=\{h_n\}$ be an ON basis of $\hi$ such that it diagonalizes $\rho$, that is,
$$
\rho=\sum_{n=1}^{\dim\hi}\la_n\kb{h_n}{h_n},\qquad\la_n\ge0,\quad\sum_{n=1}^{\dim\hi}\la_n=1.
$$
Now, for each positive $B\in\lk$, 
\begin{eqnarray*}
\int_X\tr{\rho_xB}w_\rho(x)\d\mu(x)&=&\tr{\rho\M(X,B)}
=\sum_{n=1}^{\dim\hi}\lambda_n\<h_n|\M(X,B)h_n\> \\
&=&
\sum_{n=1}^{\dim\hi}\lambda_n \int_X \sum_{k=1}^{n(x)} \<\A_k(x)h_n|B\A_k(x)h_n\>\d\mu(x)
\end{eqnarray*}
for all $X\in\Sigma$, so that 
$$
\tr{\rho_xB}w_\rho(x)=\sum_{n=1}^{\dim\hi}\lambda_n\sum_{k=1}^{n(x)} \<\A_k(x)h_n|B\A_k(x)h_n\>
$$
for $\mu$-almost all $x\in\Omega$ by the monotonic convergence theorem. 
For $B=I_\ki$, one gets 
$$
w_\rho(x)=\sum_{n=1}^{\dim\hi}\lambda_n \sum_{k=1}^{n(x)} \|\A_k(x)h_n\|^2
$$
for $\mu$-almost all $x\in\Omega$.
Since any $B\in\lh$ can be decomposed into positive parts, 
one sees that, if $w_\rho(x)\ne 0$, then the posterior state corresponding to the `outcome' $x\in\Omega$ can be chosen to be
$$
\rho_x=w_\rho(x)^{-1}\sum_{n=1}^{\dim\hi}\lambda_n \sum_{k=1}^{n(x)}\kb{\A_k(x)h_n}{\A_k(x)h_n}=
w_\rho(x)^{-1}\sum_{k=1}^{n(x)}\A_k(x)\rho\A_k(x)^*\qquad\text{(formally)}.
$$
This is a generalization of Holevo's result \cite{Ho98}.
Sometimes (in the case $\ki=\hi$) $\rho_x$ is interpreted as a state conditioned upon an outcome $x$ of the measurement described by $\M$, that is,
this measurement of the POVM $\Mo_{I_\hi}$ changes the initial state $\rho$ into a posterior state $\rho_x$ if the value $x$ is observed \cite{Oz85,Ho98}.
This intrepretation is problematic since, on the first hand, $\rho_x$ is not necessarily unique. On the other hand, 
even if $\{x\}\in\Sigma$ it may happen that $\mu(\{x\})=0$ (e.g.\ position observables). Then it is better to
define a {\it conditional output state} 
$$
\rho_X=\M_*(X,\rho)/\tr{\M_*(X,\rho)}=\mu^\M_\rho(X)^{-1}\sum_{k=1}^{n(X)}\A_k(X)\rho\A_k(X)^*\in\mathcal S(\ki)
$$
{\it corresponding to a set} $X\in\Sigma$ of outcomes (if $\tr{\M_*(X,\rho)}>0$), which describes the state, at the instant after the measurement, of the subensemble of the measured system in which the outcomes of the measurement lie in $X$.
However, even if $X\subset X'$ it often happens that $\rho_{X}\ne\rho_{X'}$. %Hence, the choice of registering the values in $X$ is trivially also the choice of registering them in the larger set $X'$ and the conditional states should be the same!

Notice that
$$
\rho_X=\mu^\M_\rho(X)^{-1}\int_X\rho_x\d\mu^\M_\rho(x)
$$
so that $\rho_{\{x\}}=\rho_x$ if and only if $\{x\}\in\Sigma$ and $\mu^\M_\rho(\{x\})>0$.
Hence, in the discrete case, there is no above mentioned problems. The most familiar example is of course 
the (strongly repeatable) von Neumann-L\"uders instrument
$$
\M(X,B)=\sum_{i\in X}P_iBP_i
$$
where $\{P_i\}$ constitute a discrete PVM (see, e.g.\ \cite{DaLe,Da,Oz84,BuLaMi}). Now
 if the value $i$ is observed then the state at the instant after the measurement is 
$$
\rho_i=\rho_{\{i\}}=\frac{P_i\rho P_i}{\tr{\rho P_i}}.
$$
Finally, we recall that the von Neumann-L\"uders instrument is extreme \cite{DAPeSe}.

%von Neumann: "PVM can be measured with absolute precision if and only if its spectrum is discrete" \cite{DaLe}
%von Neumann: Discrete  PVM $\Mo$. 

\section{Refinements and rank 1 instruments}

Let $\M\in\In(\Sigma,\ki,\hi)$, and $\hd=\int_\Omega^\oplus\hi_{n(x)}\d\mu(x)$ the direct integral and operators $\A_k(x)$ the pointwise Kraus operators of Theorem \ref{th2} associated to $\M$.
Let $${\bf N}:=\{k\in\N_\infty\,|\,1\le k \le\dim\hi\dim\ki\}$$
and $\#:\,2^{\bf N}\to\N_\infty$ be the counting measure, i.e.\ $\#N$ is the number of the elements of $N\subseteq\bf N$. Let $\#\times\mu$ be the product measure defined on the product $\sigma$-algebra $\Sigma_{\bf N}\subseteq 2^{{\bf N}\times\Omega}$ of $2^{\bf N}$ and $\Sigma$. 
Define $\A(k,x):=\A_k(x)$ when $k\le n(x)$ and $\A(k,x):=0$ for $k>n(x)$. Then, for all $\fii,\,\psi\in V_\e$,
 \begin{eqnarray*}
\<\fii|\M(X,B)\psi\>
&=& \int_X \sum_{k=1}^{n(x)} \<\fii|\A_k(x)^*B\A_k(x)\psi\>\d\mu(x) \\
&=& \int_{{\bf N}\times X} \<\fii|\A(k,x)^*B\A(k,x)\psi\>\d(\#\times\mu)(k,x) \\
\end{eqnarray*}
and we can define the {\it maximally refined} (rank-1) instrument $\M_1\in\In(\Sigma_{\bf N},\ki,\hi)$ of $\M$ by
$$
\<\fii|\M_1(\ov X,B)\psi\>:=\int_{\ov X} \<\fii|\A(k,x)^*B\A(k,x)\psi\>\d(\#\times\mu)(k,x),\quad\fii,\,\psi\in V_\e,\;\ov X\in\Sigma_{\bf N},\;B\in\lk.
$$
\begin{proposition}\label{jgjhfjoooot}
If $\M\in\ext\In(\Sigma,\,\ki,\,\hi)$
then $\M_1\in\ext\In(\Sigma_{\bf N},\,\ki,\,\hi)$.
\end{proposition}
\begin{proof} 
Let $\M\in\In(\Sigma,\,\ki,\,\hi)$ be extreme.
From Remark \ref{exrem} one sees that,
for any decomposable operator $D\in\li(\hd)$, the condition
$$
\int_\Omega\sum_{k,l=1}^{n(x)}D(x)_{kl}\<\fii|\A_k(x)^*\A_l(x)|\psi\>\d\mu(x)=0,
\hspace{0.5cm} \fii,\,\psi\in V_\e,
$$
implies $D=0$. Especially, by choosing all operators $D(x)$ diagonal, i.e.,
$D(x)_{kl}=\delta_{kl}d(k,x)$ for all $x\in\Omega$ 
(where $d\in L^\infty(\#\times\mu)$ is such that 
$\|D\|={\rm ess\,sup}\{|d(k,x)|\,|\,x\in\Omega,\,k< n(x)+1\}<\infty$)
 we get that the condition
$$
\ \int_{{\bf N}\times \Omega}d(k,x)\<\fii|\A(k,x)^*\A(k,x)\psi\>\d(\#\times\mu)(k,x)=
\int_\Omega\sum_{k,l=1}^{n(x)}\delta_{kl}f_k(x)\<\fii|\A_k(x)^*\A_l(x)|\psi\>\d\mu(x)=0
$$
implies $d=0$. Hence,
$\M_1\in\ext\In(\Sigma_{\bf N},\,\ki,\,\hi)$ by Remark \ref{exrem}.
\end{proof}

As a special case ($\ki=\C$), one sees that any POVM can be maximally refined (recall also Example \ref{ExPOVM}).
Thus, next we consider the instruments and measuring processes of rank-1 POVMs. Note that important examples of rank-1 POVMs are position and momentum observables (of a spin 0 particle moving on a space manifold), rotated quadratures, phase space observables generated by pure states, the canonical phase observable of a single mode electromagnetic field, and many important discrete observables. Often their related instruments are also of rank-1 (see, e.g.\ Section 4.6 of \cite{Da}).

Let $\Mo\in\O(\Sigma,\hi)$ be a rank-1 POVM, that is, in Theorem \ref{th1}, $n(x)\in\{0,1\}$ for all $x\in\Omega$, $\hi_1\cong\C$, $\hi_0=\{0\}$,
$\hd=L^2(\mu)$, and for all $X\in\Sigma$ we have
$$
\Mo(X)=
\sum_{n,m=1}^{\dim\hi}\int_{X}\ov{\psi_n(x)}\psi_m(x)\d\mu(x)\kb{h_n}{h_m} 
$$
or
$$
\<\fii|\Mo(X)\psi\>=\int_X  \<\fii|d_1(x)\>\<d_1(x)|\psi\>\d\mu(x),\hspace{0.5cm}\fii,\,\psi\in V_\e.
$$
Let $\M\in\In(\Sigma,\ki,\hi)$. Then, by Theorem \ref{compa},
$\M(X,I_\ki)\equiv\Mo(X)$ if and only if
there exists a decomposable CP channel $T:\,\lk\to\li\big(L^2(\mu)\big)$,
$B\mapsto T(B)=\int_\Omega^\oplus\tr{\sigma_xB}\d\mu(x)$, where $\sigma_x\in\mathcal S(\ki)$ are states for $\mu$-almost all $x\in\Omega$, such that
$$
\M(X,B)\equiv 
\sum_{n,m=1}^{\dim\hi}\int_{X}\tr{\sigma_xB}\ov{\psi_n(x)}{\psi_m(x)}\d\mu(x)\kb{h_n}{h_m}
=\int_X\tr{\sigma_xB}\d\Mo(x).
$$
Any channel $B\mapsto\tr{\sigma_xB}$ above can be chosen to be of the form (see, Theorem \ref{compa})
$$
\tr{\sigma_xB}=\big\<\eta_x\big|(B\otimes I_{\hi'_{n'(x)}})\eta_x\big\>,\qquad B\in\lk,
$$
where $\eta_x\in \ki\otimes\hi'_{n'(x)}$ is a unit vector and
$\hi'_\oplus=\int_\Omega^\oplus\hi'_{n'(x)}\d\mu(x)$ is the direct integral Hilbert space of Theorem \ref{th2} associated to $\M$. 
Hence, we have shown that {\it any compatible instrument of a rank-1 POVM is a nuclear instrument} (see, Example \ref{nuclearex}) and thus can be identified with (the equivalence class of) $\{\sigma_x\}_{x\in\Omega}$.
By writing (for $\mu$-almost all $x\in\Omega$)
$$
\eta_x=\sum_{k=1}^{n'(x)}\fii_k(x)\otimes b_k
$$
we see that
$$
\tr{\sigma_xB}
=\sum_{k=1}^{n'(x)}\<\fii_k(x)|B\fii_k(x)\>,\qquad
\sigma_x=\sum_{k=1}^{n'(x)}\kb{\fii_k(x)}{\fii_k(x)}
$$
and $n'(x)$ is the rank of the state $\sigma_x$.
Hence, $\M$ is of rank-1 (i.e.\ $n'(x)\in\{0,1\}$) if and only if $\sigma_x$ is a pure state for $\mu$-almost everywhere.
In this case one may choose $\hi'_\oplus=\hd$.

One can also use another decomposition
$$
\eta_x=\sum_{s=1}^{\dim\ki}k_s\otimes\eta^s(x),\qquad
\sum_{s=1}^{\dim\ki}\|\eta^s(x)\|^2=1,
$$
to get 
$$
\tr{\sigma_xB}=\sum_{s,t=1}^{\dim\ki}B_{st}\<\eta^s(x)|\eta^t(x)\>,
\qquad \sigma_x=\sum_{s,t=1}^{\dim\ki}\<\eta^s(x)|\eta^t(x)\>\kb{k_t}{k_s},
$$
and
$$
\M(X,B)\equiv
\sum_{n,m=1}^{\dim\hi}
\sum_{s,t=1}^{\dim\ki} 
\int_{X}B_{st}\<\psi_n(x)\eta^s(x)|\psi_m(x)\eta^t(x)\>\d\mu(x)\kb{h_n}{h_m}
$$
so that $\psi^s_n(x)=\psi_n(x)\eta^s(x)\in\hi'_{n'(x)}$, 
$\sum_s k_s\otimes\psi_n ^s(x)=\psi_n(x)\eta_x$,
and the generalized vectors of $\M$ can be chosen to be
$$
d_k^s(x)=c^x_{1,sk}d_1(x),\qquad c^x_{1,sk}=\<\eta^s(x)|b_k\>,\qquad\sum_{s,k}|c^x_{1,sk}|=1.
$$
If $\M$ is also of rank 1 then $\eta_x\in\ki\cong\ki\otimes\hi'_{1}$,
$\eta^s(x)\in\C$, and $\sigma_x=\kb{\eta_x}{\eta_x}$ for $\mu$-almost everywhere.
Note that, following Holevo \cite{Ho08}, one sees that {\it any EB channel can be seen as a rank-1 nuclear instrument whose associate observable is of rank-1.}

\begin{theorem}\label{hnasgcbnsjcd}
Let $\Mo\in\O(\Sigma,\hi)$ be a rank-1 POVM and $\M\in\In(\Sigma,\ki,\hi)$ any $\Mo$-compatible instrument. 
Then $\M$ is nuclear.
If $\M\in\ext\In(\Sigma,\ki,\hi)$ then $\Mo\in\ext\O(\Sigma,\hi)$.
If $\M$ is also of rank 1, then $\Mo\in\ext\O(\Sigma,\hi)$ implies 
$\M\in\ext\In(\Sigma,\ki,\hi)$.
\end{theorem}

\begin{proof}
Since $\M\in\ext\In(\Sigma,\ki,\hi)$ if and only if the condition 
\begin{equation}\label{hfheivndkibvhdfidj}
\sum_{s=1}^{\dim\ki} \int_\Omega\<\psi_n^s(x)|D(x)\psi_m^s(x)\>\d\mu(x)=\int_\Omega\sum_{s=1}^{\dim\ki}
\<\eta^s(x)|D(x)\eta^s(x)\>\ov{\psi_n(x)}\psi_m(x)\d\mu(x)
=0
\end{equation}
for all $n,\,m$,
implies $D=0$ (see, Remark \ref{exrem}).
Let $d\in L^\infty(\mu)$ and set $D(x)=d(x)I_{\hi'_{n'(x)}}$ to get that
$$
\int_\Omega d(x)\d\Mo(x)=0
$$
implies $d=0$, i.e.\ $\Mo$ is extreme. 

From \eqref{hfheivndkibvhdfidj}, one sees that
if $\Mo$ and (an $\Mo$-compatible) $\M$ are of rank 1 then their extremality conditions are exactly the same condition:
$\int_\Omega d(x)\d\Mo(x)=0$ (where $d\in L^\infty(\mu)$) implies $d=0$.  
\end{proof}

%$$
%\M_*(X,\rho)=\sum_{s,t=1}^{\dim\ki} 
%\sum_{n,m=1}^{\dim\hi}
%\int_{X}\rho_{mn}\<\eta^s(x)\otimes \psi_n(x)|\eta^t(x)\otimes \psi_m(x)\>\d\mu(x)\kb{k_t}{k_s}
%$$
%
%
%
%
%\begin{align*}
%\M_*(X,\rho)={\rm tr}_2\big[U(\rho\otimes\sigma)U^*\big(I_\hi\otimes\CHII X\big)\big]
%=
%\sum_{n,m=1}^{\dim\hi}\rho_{mn}
%{\rm tr}_2\big[U(\kb{h_m}{h_n}\otimes\sigma)U^*\big(I_\hi\otimes\CHII X\big)\big]
%\end{align*}
%
%$$
%{\rm tr}_2\big[U(\kb{h_m}{h_n}\otimes\sigma)U^*\big(I_\hi\otimes\CHII X\big)\big]
%=
%\sum_{s,t}
%\<\eta^s\otimes \psi_n|\CHII X\eta^t\otimes \psi_m\>\kb{h_t}{h_s}
%$$
%
%$$
%U(\kb{h_m}{h_n}\otimes\sigma)U^*
%=\sum_{s,t}\kb{h_t}{h_s}\otimes\kb{\eta^t\otimes \psi_m}{\eta^s\otimes \psi_n}
%=\sum_{s,t}\kb{h_t\otimes\eta^t\otimes \psi_m}{h_s\otimes\eta^s\otimes \psi_n}
%$$
%
%$$
%U(x)(h_n\otimes SHIT)=\sum_s h_s\otimes\eta^s(x)\otimes\psi_n(x)
%$$
%yleisesti
%$$
%U(x)(h_n\otimes SHIT)=\sum_s h_s\otimes\psi_n^s
%$$

\section{Modules over $C^*$-algebras and extreme kernels and CP-maps}

In this section, we follow \cite{PeYl}.
We let $A$ and $A_e$ be $C^*$-algebras and assume that $A_e$ is unital with the unit $e\in A_e$.

If $M$ is a (right) Hilbert $C^*$-module over $A$ and $B:\,M\to M$ a bounded $A$-linear\footnote{We assume implicitly that $A$-(sesqui)linear maps are also $\C$-(sesqui)linear.} map then $B$ is a positive element of the $C^*$-algebra $L(M)$ of adjointable maps on $M$ (i.e.\ $B\ge 0$) 
if and only if $\<v|Bv\>\ge 0$ for all $v\in M$
if and only if $B=C^*C$ where $C:\,M\to M'$ is adjointable (and $M'$ a Hilbert $C^*$-module over $A$)
\cite[Proposition 2.1.3]{Manuilov}. We denote $B\le B'$ if $B'-B\ge 0$, and let $I_M$ be the identity (operator) of $L(M)$.

Let $n\in\N_+$ and $M_n(A)$ be the matrix
$C^*$-algebra consisting of the $A$-valued $n\times n$--matrices
$(a_{ij})_{i,j=1}^n$. 
An element $(a_{ij})_{i,j=1}^n$ of $M_n(A)$ is
positive if and only if $\sum_{i,j=1}^n a_i^*a_{ij} a_j\ge 0$ for
all $a_1,\ldots,a_n\in A$. 
%Note that, for any positive matrix $(a_{ij})_{i,j=1}^n$, $a_{ij}^*=a_{ji}$ and $\sum_{i,j=1}^n a_{ij}\ge 0$.

Let $V$ be an $A$-module (i.e.\ a right module over the algebra $A$) and $S_A(V)$ the $\C$-linear space of $A$-sesquilinear maps $s:\,V\times V\to A$. Let
$M_n\big(S_A(V)\big)$ be the
$\C$-linear space of $n\times n$--matrices $(s_{ij})_{i,j=1}^n$ where
the matrix elements $s_{ij}$ belong to 
to $S_A(V)$. Note that the matrix multiplication is not defined.
We say that  $(s_{ij})_{i,j=1}^n\in M_n\big(S_A(V)\big)$ is {positive} if the matrix $\big(s_{ij}(v_i,v_j)\big)_{i,j=1}^n$ is a positive element of $M_n(A)$ for all $v_1,\dots,v_n\in V$.

\subsection{Positive-definite kernels}

For any set $X\ne\emptyset$, we say that a mapping $K:\,X\times X\to
S_A(V)$ is a {\it positive-(semi)definite $A$-kernel} if, for all $n\in\N_+$
and $x_1,\ldots,x_n\in X$, the matrix $\big(K(x_i,x_j)
\big)_{i,j=1}^n\in M_n\big(S_A(V)\big)$ is positive, i.e.\ for all $v_1,\dots,v_n\in V$
$$
\sum_{i,j=1}^n\big[K(x_i,x_j)\big](v_i,v_j)\ge 0.
$$
Let $\K$ denote the convex set of all positive-definite $A$-kernels $K:\,X\times X\to S_A(V)$.
The following theorem is proved in \cite{PeYl}:
\begin{theorem}\label{t1}
For each $K\in\K$ there exists a Hilbert $C^*$-module $M$ over $A$ and $A$-linear maps
$D(x):\,V\to M$,  $x\in X$, such that
\begin{itemize}
\item[(i)]
$\big[K(x,x')\big](v,v')=\<D(x)v|D(x')v'\>,\hspace{5mm}x,\,x'\in
X,\;v,\,v'\in V, $
\item[(ii)]
$\lin_\C\cup_{x\in X}D(x)V$ is dense in $M$.
\end{itemize}
We say that $(M,D)$ is a {\em minimal Kolmogorov decomposition (MKD)} for $K$. %and the complex dimension of $M$ is the {\em rank} of $K$, ${\rm rank}\,K:=\dim M$.

If $(M',D')$ is another MKD for $K$
then there exists a unitary $U:\,M\to M'$ such that $UD(x)=D'(x)$
for all $x\in X$.
\end{theorem}

Any $K\in\K$ is said to be {\it regular} if the Hilbert $C^*$-module $M$ (associated 
with a MKD $(M,D)$ of $K$)  is self-dual. Note that {\it any} $K\in\K$ is regular if $A$ is finite dimensional (e.g.\ when $A\subseteq\lh$ where $\hi$ is a finite dimensional Hilbert space) \cite[Section 2.5]{Manuilov}.

Fix $K_1\in\K$ and $Z\subseteq X\times X$. Let
${\bf C}(K_1,Z)\subseteq\K$ consist of positive-definite $A$-kernels $K:\,X\times X\to S_A(V)$ such that $K(x,x')=K_1(x,x')$ for all $(x,x')\in Z$. Obviously, ${\bf C}(K_1,Z)$ is convex and we denote by $\ext {\bf C}(K_1,Z)$ its extreme points.
Note that ${\bf C}(K_1,\emptyset)=\K$.

\begin{theorem}\label{extker}
%Let $Z\subseteq X\times X$ be symmetric and 
Let $K\in {\bf C}(K_1,Z)$ be regular and $(M,D)$ its MKD.
Then $K\in\ext {\bf C}(K_1,Z)$ if and only if, for all selfadjoint $B\in L(M)$, 
%for which $\<D(x)v|BD(x')v'\>=0$ for all $(x,x')\in Z$ and $v,\,v'\in V$,
the condition 
$$
\<D(x)v|BD(x')v'\>=0,\hspace{0.5cm}(x,x')\in Z,\;v,\,v'\in V,
$$ 
implies $B=0$. %Especially, rank 1 kernels are extreme.
\end{theorem}

\begin{proof}
For any $J,\,K\in\K$ we denote $J\le K$ if $K-J\in\K$. 
Let $K\in {\bf C}(K_1,Z)$ be regular and $(M,D)$ its MKD.
Suppose that $$K=\frac12 K_++\frac12 K_-$$ where $K_\pm\in {\bf C}(K_1,Z)$ and $K_+\ne K_-$.
Then $K_\pm\le 2 K$.
Let $(M_\pm ,D_\pm )$ be a MKD of $K_\pm $ so that, for all 
$n\in\N_+$, $x_1,\ldots,x_n\in X$, and $v_1,\dots,v_n\in V$,
$
\sum_{i,j=1}^n\big[K_\pm (x_i,x_j)\big](v_i,v_j) \le 2\sum_{i,j=1}^n\big[K(x_i,x_j)\big](v_i,v_j)
$
or, equivalently,
$$
\left\<\sum_{i=1}^nD_\pm (x_i)v_i\Bigg|\sum_{j=1}^nD_\pm (x_j)v_j\right\> 
\le 2
\left\<\sum_{i=1}^nD(x_i)v_i\Bigg|\sum_{j=1}^nD(x_j)v_j\right\>.
$$
Hence, we may define $\C$-linear maps $G_\pm :\,M\to M_\pm $ by
$$
G_\pm \big(D(x)v\big):=D_\pm (x)v,\hspace{0.5cm}x\in X,\;v\in V,
$$
for which
$
\<G_\pm w|G_\pm w\>\le2\<w|w\>
$
for all $w\in\lin_\C\cup_{x\in X}D(x)V\subseteq M$. It follows that $G_\pm $ is well defined, $A$-linear, and bounded (with the norm $\|G_\pm \|\le\sqrt2$) \cite[Theorem 2.1.4, Corollary 2.1.6]{Manuilov}.
By regularity of $K$ the module $M$ is self-dual and thus $G_\pm $ is adjointable with the adjoint $G_\pm ^*:\,M_\pm \to M$ \cite[Proposition 2.5.2]{Manuilov}. Define a positive $B_\pm :=G_\pm ^*G_\pm\in L(M)$ so that
$$
\big[K_\pm (x,x')\big](v,v')=\<G_\pm D(x)v|G_\pm D(x')v'\>=\<D(x)v|G_\pm ^*G_\pm D(x')v'\>=\<D(x)v|B_\pm D(x')v'\>
$$
for all $x,\,x'\in X$ and $v,\,v'\in V$. Let $B:=B_+-B_-\in L(M)$ for which $B^*=B$. Since $K_\pm \in {\bf C}(K_1,Z)$ one has
\begin{eqnarray*}
\<D(x)v|BD(x')v'\>&=&\<D(x)v|B_+D(x')v'\>-\<D(x)v|B_-D(x')v'\> \\
&=&\big[K_1(x,x')\big](v,v')-\big[K_1(x,x')\big](v,v')=0 
\end{eqnarray*}
for all $(x,x')\in Z$ and $v,\,v'\in V$.
Since $B=0$ if and only if $B_+= B_-$ if and only if $K_+= K_-$ it follows that one must have $B\ne 0$.

Suppose then that there exists a nonzero $B\in L(M)$, $B^*=B$, which satisfies the condition of the theorem.
We may assume that $\|B\|\le 1$ (otherwise redefine $B$ to be $\|B\|^{-1}B$).
Since $\pm B\le \|B\| I_M$ it follows that
$$
B_\pm:=I_M\pm B\ge 0,\hspace{1cm}B_+\ne B_-.
$$
Define $K_\pm\in\K$ by
$$
\big[K_\pm(x,x')\big](v,v'):=\<D(x)v|B_\pm D(x')v'\>,\hspace{5mm}x,\,x'\in
X,\;v,\,v'\in V, 
$$
for which $K_+\ne K_-$, $K=\frac12 K_++\frac12 K_-$, and 
$$
\big[K_\pm(x,x')\big](v,v')=
\<D(x)v|D(x')v'\>
\pm\underbrace{\<D(x)v|BD(x')v'\>}_{=\;0}
=\big[K_1(x,x')\big](v,v')
$$
for all $(x,x')\in Z$ and $v,\,v'\in V$. Hence, $K_\pm\in{\bf C}(K_1,Z)$ and $K$ is not extreme.
\end{proof}

Note that, if the set $Z\subseteq X\times X$ of the preceding theorem is symmetric (i.e.\ $(x,x')\in Z$ implies $(x',x)\in Z$)
then the condition applied for {\it any} $B\in L(M)$ implies $B=0$ if and only if $K$ is extreme (if $B^*\ne B$ redefine $B$ to be $i(B-B^*)$ which also satisfies the condition since $Z$ is symmetric).

\begin{example}[Autocorrelation functions] \rm
Let $A=\C$ and $K\in\K$ with a MKD $(M,D)$ so that
\begin{itemize}
\item $V$ is a vector space (i.e.\ a $\C$-module),
\item $S_\C(V)$ consists of sesquilinear forms $V\times V\to\C$,
\item $K$ is regular (since $\dim\C=1<\infty$),
\item $M$ is a Hilbert space (i.e.\ a Hilbert $C^*$-module over $\C$),
\item any $D(x):\,V\to M$ is linear.
\end{itemize}
Suppose further that $V=\C$ so that 
$$
\big[K(x,x')\big](c,c')=\<D(x)c|D(x')c'\>=\ov c\<D(x)1|D(x')1\>c',\hspace{5mm}x,\,x'\in
X,\;c,\,c'\in\C, 
$$
and $K$ can be identified with the positive definite $\C$-kernel 
$$
k:\,X\times X\to\C,\;(x,x')\mapsto k(x,x'):=\big[K(x,x')\big](1,1)=\<m(x)|m(x')\>
$$
where $m:\,X\to M,\,x\mapsto m(x):=D(x)1$. We say that $(M,m)$ is a MKD of $k$ (i.e.\ the set of the linear combinations of vectors $m(x)$ is dense in the Hilbert space $M$).
Let $K_1\in {\bf K}(X,\C)$ be defined by a MKD $(M_1,D_1)$ where
$M_1=\C$ and $D_1(x):\,\C\to\C,\,c\mapsto c,$ for all $x\in X$. Then the corresponding positive semidefinite function is the constant function $k_1(x,x')\equiv1$.
If $Z=\{(x,x)\in X\times X\,|\,x\in X\}$ then ${\bf C}(K_1,Z)$ can be identified with the convex set ${\bf C}_1(X)$ of positive semidefinite functions $k:\,X\times X\to\C$ with the unit diagonal (i.e.\ $k(x,x)\equiv1$).
In the context of stochastic processes, such a $k$ is called an {\it autocorrelation function}, and the characterization of extreme autocorralation functions has long been a problem (see, e.g.\ \cite{LiTa,KiPe} and references therein).
The next immediate corollary of Theorem \ref{extker} solves this problem completely.
\begin{proposition}\label{auto}
Let $k\in{\bf C}_1(X)$ and $(M,m)$ its MKD. Then $k$ is extreme in ${\bf C}_1(X)$ if and only if, 
for any bounded operator $B:\,M\to M$, the condition
$$
\<m(x)|Bm(x)\>=0,\hspace{0.5cm} x\in X,
$$
implies $B=0$.
\end{proposition}
The above proposition is a generalization of Theorem 1 of \cite{KiPe} and hence a generalization of Theorem 1 of \cite{LiTa}. Following \cite{LiTa,KiPe}  we see that there exist an extreme autocorrelation function $m$ of any rank $r:=\dim M\le\sqrt{\#X}$.
\end{example}

\subsection{Completely positive maps}

Let $E:\,A_e\to S_A(V)$ be a $\C$-linear mapping. 
For any $n\in\N_+$
we define the $n^{\rm th}$ {amplification} $E^{(n)}:\,M_n(A_e)\to
M_n\big(S_A(V)\big)$ of $E$ as
$E^{(n)}\big((b_{ij})_{i,j}\big):=\big(E(b_{ij})\big)_{i,j}$. For
example, $E^{(1)}=E$. We say that $E^{(n)}$ is positive if
$E^{(n)}\big((b_{ij})_{i,j}\big)$ is positive for any positive
$(b_{ij})_{i,j}$. Moreover, $E$ is {\it completely positive (CP)} if $E^{(n)}$ is positive for all $n\in\N_+$.
It is easy to see that $E$ {\it is CP if and only if $\tilde E:\,A_e\times A_e\to
S_A(V)$, $(b,b')\mapsto E(b^*b')$ is a positive-definite $A$-kernel,} i.e.\ $$\sum_{i,j=1}^n E(b_i^*b_j)(v_i,v_j)\ge 0$$ for all $n\in\N_+$, $b_1,\ldots,b_n\in A_e$, and $v_1,\ldots,v_n\in V$.

We have proved the following generalization of Kasparov-Stinespring-Gelfand-Naimark-Segal (KSGNS) theorem \cite{PeYl}:

%Denote by $L_A(M)$ is the set of adjointable (and thus $A$-linear and bounded) maps on a Hilbert module $M$.

\begin{theorem}\label{kuusi}
Let $E:\,A_e\to S_A(V)$ be a CP ($\C$-linear) map.
There exist a Hilbert $C^*$-module $M$ over $A$, a {unital}
*-homomorphism $\pi:\,A_e\to L(M)$, and an $A$-linear $Y:\,V\to M$ such that
\begin{itemize}
\item[(i)]
$ E(b)(v,v')=\<Yv|\pi(b)Yv'\>,
\hspace{5mm}b\in A_e,\;v,\,v'\in V, $
\item[(ii)]
$\lin_\C\,\pi(A_e)YV=\lin_\C \{ \pi (b)Yv\, | \, b\in A_e,\,v\in V\}$ is
dense in $M$.
\end{itemize}
We say that $(M,\pi,Y)$ is a {\em minimal dilation} for $E$.

If $(M',\pi',Y')$ is another minimal dilation for $E$ then there is
a unitary mapping $U:\,M\to M'$ such that
$$
\pi'(b)=U\pi(b)U^*, \hspace{5mm}b\in A_e,
$$
and $Y'=UY$.
\end{theorem}

Note that a MKD $(M,D)$ of $\tilde E$ is related with a minimal dilation $(M,\pi,Y)$ of $E$. Indeed, as shown in \cite{PeYl}, one may choose $Y=D(e)$ and $\pi(b)Y=D(b)$ so that $\big[\tilde E(b,b')\big](v,v')=\<D(b)v|D(b')v'\>=\<\pi(b)Yv|\pi(b')Yv'\>=\<Yv|\pi(b)^*\pi(b')Yv'\>=\<Yv|\pi(b^*b')Yv'\>=E(b^*b')(v,v')$.
We say that $E$ is regular if $\tilde E$ is regular.

Denote by ${\bf CP}(A_e,V)$ the convex set of CP maps $E:\,A_e\to S_A(V)$ and fix an $E_1\in{\bf CP}(A_e,V)$ and $Z=\{(b,b')\in A_e\times A_e \, | \,b^*b'=e\}$. Denote $s_e:=E_1(e)\in S_A(V)$ and define a convex set
$$
{\bf CP}(A_e,V,s_e):=\{ E \in {\bf CP}(A_e,V) \, | \, E(e)=s_e\}\subseteq{\bf C}(\tilde E_1,Z)
$$
From Theorem \ref{extker} one gets that a regular $E\in{\bf CP}(A_e,V,s_e)$ is extreme in ${\bf C}(\tilde E_1,Z)$ if and only if,
for all $B\in L(M)$, 
the condition 
$$
\<\pi(b)Yv|B\pi(b')Yv'\>=0,\hspace{0.5cm}b,\,b' \in A_e,\;b^*b'=e,\;v,\,v'\in V,
$$ 
implies $B=0$. The next theorem characterizes completely the regular extreme points of the smaller set ${\bf CP}(A_e,V,s_e)$ and is thus a generalization of \cite[Theorem 1.4.6]{Ar}:

\begin{theorem}\label{seiska}
Let $E\in{\bf CP}(A_e,V,s_e)$ be regular and $(M,\pi,Y)$ its minimal dilation.
Then $K\in\ext{\bf CP}(A_e,V,s_e)$ if and only if, for all $B\in L(M)$ commuting with $\pi$ (i.e.\ $[B,\pi(b)]=0$ for all $b\in A_e$)
the condition 
$$
\<Yv|BYv'\>=0,\hspace{0.5cm}v,\,v'\in V,
$$ 
implies $B=0$. %Especially, rank 1 CP-maps are extreme.
\end{theorem}

\begin{proof}
We follow the proof of Theorem \ref{extker} and use the same notations defined there:
Assume that $E=\frac12 E_++\frac12 E_-$, $E_\pm\in{\bf CP}(A_e,V,s_e)$, $E_+\ne E_-$, and denote
$K=\tilde E$ and $K_\pm=\tilde E_{\pm}$ (so that $\tilde E=\frac12\tilde E_++\frac12\tilde E_-$).
Now 
$$
E_\pm(b^*b')(v,v')=
\big[\tilde E_\pm (b,b')\big](v,v')=\<\pi(b)Yv|B_\pm \pi(b')Yv'\>,%=\<Yv|\pi(b^*)B_\pm \pi(b')Yv'\>,
\hspace{0.5cm}b,\,b'\in B,\;v,\,v'\in V,
$$
where $B_\pm\in L(M)$ are positive, and one gets
$$
E_\pm (b^*b''b')(v,v')=\big[\tilde E_\pm ((b'')^*b,b')\big](v,v')=\big[\tilde E_\pm (b,b''b')\big](v,v'),
$$
that is,
$$
\<\pi(b)Yv|(B_\pm\pi(b'')-\pi(b'')B_\pm)\pi(b')Yv'\>=0,\hspace{0.5cm}b,\,b',\,b''\in B,\;v,\,v'\in V.
$$
By density, $[B_\pm,\pi(b'')]=0$ for all $b''\in B$ and 
$B:=B_+-B_-\ne 0$ commutes with $\pi$.
Moreover, $\<Yv|BYv'\>=E_+(e)(v,v')-E_-(e)(v,v')=s_e(v,v')-s_e(v,v')=0$.

Suppose then that there exists a nonzero $B\in L(M)$ which satisfies the condition of the theorem.
We may assume that $B^*=B$ and $\|B\|\le 1$ so that 
$B_\pm:=I_M\pm B\ge 0$, $B_+\ne B_-$, and $B_\pm$ commute with $\pi$.
Define $E_\pm\in{\bf CP}(A_e,V,s_e)$ by
$$
E_\pm(b)(v,v'):=\<Yv|B_\pm \pi(b)Yv'\>,\hspace{5mm}b\in B,\;v,\,v'\in V, 
$$
for which $E_+\ne E_-$ and $E=\frac12 E_++\frac12 E_-$
so that $E$ is not extreme.
\end{proof}

\begin{example}[Instruments] \rm \label{instruex}
Let $\M:\,\Sigma\times \lk\to \lh$ be a CP instrument and $\mu:\,\Sigma\to[0,\infty]$ a $\sigma$-finite measure such that $\M$ is absolutely continuous with respect to it. Denote, by the same symbol $\M$, the CP `extension' (or linearization) of $\M$ to the tensor product (von Neumann) algebra 
$\lk\otimes L^\infty(\mu)$ 
of von Neumann algebras $\lk$ and $L^\infty(\mu)$
(with the preduals $\tk$ and $L^1(\mu)$, respectively).
Recall that $\lk\otimes L^\infty(\mu)$ 
 can be viewed as the (equivalence classes of) $\mu$-measurable field of operators,  
$\Omega\ni x\mapsto B(x)\in\lk$, such that $\|B\|:=\mu\text{-ess sup}_{x\in\Omega}\|B(x)\|<\infty$. %Its predual consists of $\mu$-measurable fields $\Omega\ni x\mapsto \rho(x)\in\tk$ for which $\|\rho\|_1:=\int_\Omega\tr{\sqrt{\rho(x)^*\rho(x)}}\d\mu(x)<\infty$, and it is $\lk\otimes L^\infty(\mu)$-module.

In Theorems \ref{kuusi} ja \ref{seiska}, $E=\M$, $A_e=\lk\otimes L^\infty(\mu)$, $A=\C$, $V=\hi$, $s_e=I_\hi$, $M$ is a Hilbert space, $Y:\,\hi\to M$ is a linear isometry, and 
$
\pi:\,\lk\otimes L^\infty(\mu)\to\li(M)
$
a unital ${}^*$-homomorphism. As shown in the proof of Theorem \ref{th2},
$M$ can be chosen to be $M=\ki\otimes\hd$ and 
$\pi(B\otimes f)=B \otimes \hat f$.

For Theorem \ref{th2} one gets an interesting observation:
Let 
\begin{eqnarray*}
\<\fii|\M(X,B)\psi\> = \int_X \sum_{k=1}^{n(x)} \<\fii|\A_k(x)^*B\A_k(x)\psi\>\d\mu(x),
\hspace{0.5cm}\fii,\,\psi\in V_\e,
\end{eqnarray*}
be a minimal pointwise Kraus form of $\M$. For all $x\in\Omega$, define a rank $n(x)$ CP map
$$
\M(x,\bullet):\,\lk\to S_\C(V_\e),\;B\mapsto\M(x,B),
$$
where
$$
\M(x,B)(\fii,\psi):=\sum_{k=1}^{n(x)} \<\fii|\A_k(x)^*B\A_k(x)\psi\>,
\hspace{0.5cm}\fii,\,\psi\in V_\e,
$$
so that $\<\fii|\M(X,B)\psi\> = \int_X \M(x,B)(\fii,\psi)\d\mu(x)$ and we have obtained a CP density of a CP instrument.
Conversely, if a $\sigma$-finite positive measure $\mu$ and an ON basis $\e\subset\hi$ are given, then any (weakly $\mu$-measurable) family of CP maps 
$E_x:\,\lk\to S_\C(V_\e)$ defines an instrument if and only if $\int_\Omega E_x(I_\ki)(h_n,h_m)\d\mu(x)=\delta_{nm}$.  
\end{example}

\begin{remark}[Choi isomorphism]\rm
In this remark we generalize Choi's Theorem 2 of \cite{Ch}
which forms a basis for Choi isomorphism widely used in quantum information theory. 

Let $\k=\{k_s\}$ be a ON basis of $\ki$, $V$ a module over a $C^*$-algebra $A$, and $E:\,\lk\to S_A(V)$ a $\C$-linear map. 
Let $X:=\{s\in\N_+\,|\,s<\dim\ki+1\}$ and define a mapping $K_E:\,X\times X\to S_A(V)$ by
$K_E(s,t):=E\big(\kb{k_s}{k_t}\big)$.
We say that $K_E$ is a {\it matrix of $E$} (with respect to the basis $\k$) and denote briefly $K_E\equiv\Big(E\big(\kb{k_s}{k_t}\big)\Big)_{s,t=1}^{\dim\ki}$. Moreover, if $K_E$ is a positive-definite $A$-kernel, that is,
for all $n<\dim\ki+1$ and $v_1,\dots,v_n\in V$,
\begin{equation}\label{Choi2}
\sum_{s,t=1}^n\big[E\big(\kb{k_s}{k_t}\big)\big](v_s,v_t)\ge 0,
\end{equation}
we say that $K_E$ is {\it positive-(semi)definite matrix}. 
Equip $\lk$ with the (ultra)weak topology and $S_A(V)$ with the locally convex topology generated by the seminorms 
$s\mapsto\|s(v_1,v_2)\|$, $v_1,\,v_2\in V$.
Now we have the following generalization of \cite[Theorem 2]{Ch}:

\begin{theorem}\label{CJyleistys}
Let $E:\,\lk\to S_A(V)$ be a continuous $\C$-linear map. Then $E\in{\bf CP}(\lk,V)$ if and only if a matrix $K_E$ of $E$ is positive definite.
\end{theorem}

\begin{proof}
Now $E\in{\bf CP}(\lk,V)$ if and only if
\begin{equation}\label{zip}
\sum_{i,j=1}^n E(b_i^*b_j)(v_i,v_j)\ge 0
\end{equation}
for all $n\in\N_+$, $b_1,\ldots,b_n\in\lk$, and $v_1,\ldots,v_n\in V$.
Obviously, \eqref{Choi2} follows from \eqref{zip} by setting $b_j\equiv\kb{k_1}{k_j}$. %(and $n<\dim\ki+1$).
Suppose then that \eqref{Choi2} holds. Then \eqref{zip} clearly holds for operators $b_i$ whose matrices are finite (i.e.\ $\<k_s|b_ik_t\>\ne 0$ for finitely many indices $s$ and $t$). The ${}^*$-subalgebra $\lk_{\rm fin}$ consisting of operators with finite matrices is (ultra)weakly dense in $\lk$ so that, by continuity of $E$, the condition \eqref{zip} must hold for all operators $b_i\in\lk$.
\end{proof}

\end{remark}

\section*{Appendix}

\begin{remark}\rm \label{remu}
The vectors and operators of Theorem \ref{th2} have the following interrelations: Let $\M\in\In(\Sigma,\,\ki,\,\hi)$ be an instrument with the {\it structure vectors} $\psi_n^s\in\hd\subseteq L^2(\mu,\li)$ associated with the ON bases $\{h_n\}\subset\hi$ and $\{k_s\}\subset\ki$. Let $\{b_k\}$ be an ON basis of $\li$. Then\footnote{If $S_i:\,V_i\times V_i\to\C$, $i=1,\,2$, are sesquilinear forms (e.g.\ bounded operators) and $V_2\subseteq V_1$ (a vector subspace), we 
denote $S_1\supseteq S_2$ if $S_1(v,w)=S_2(v,w)$ for all $v,\,w\in V_2$.}
\begin{eqnarray*}
\<b_k|\psi_n^s(x)\> &=& \<d^s_k(x)|h_n\> = \<k_s|\A_k(x)h_n \>=\<b_k|\A^s(x)h_n\>, \\
\psi_n^s(x)&=&\sum_k\<d^s_k(x)|h_n\>b_k=\sum_k \<k_s|\A_k(x)h_n\>b_k=\A^s(x)h_n, \\
d^s_k(x)&=&\A_k(x)^*k_s=\A^s(x)^*b_k=\sum_{n}\<\psi_n^s(x)|b_k\>h_n,\\
\A_k(x)&=&\sum_s\kb{k_s}{b_k}\A^s(x)=\sum_{s,n}\<b_k|\psi_n^s(x)\>\kb{k_s}{h_n}=\sum_s\kb{k_s}{d^s_k(x)}, \\
\A^s(x)&=&\sum_n\kb{\psi_n^s(x)}{h_n}=\sum_k \kb{b_k}{d_k^s(x)}=\sum_k\kb{b_k}{k_s}\A_k(x), \\
(Yh_n)(x)&=&\sum_{s}k_s\otimes \psi_n^s(x)=\sum_{s,k}\<d_k^s(x)|h_n\>k_s\otimes b_k\\
&=&\sum_{k}\big(\A_k(x)h_n\big)\otimes b_k=\sum_{s}k_s\otimes\big(\A^s(x)h_n\big),\\
\M(\Omega,I_\ki)&=&I_\hi=Y^*Y\supseteq\sum_{s=1}^{\dim\ki}\int_\Omega\A^s(x)^*\A^s(x)\d\mu(x)
=\int_\Omega \sum_{k=1}^{n(x)} \A_k(x)^*\A_k(x)\d\mu(x).
\end{eqnarray*}
Note that, since $\hd\subseteq L^2(\mu,\li)$, $n(x)\le\dim\hi\dim\ki$, and $\li$ is arbitrary, it is sometimes notationally convenient to replace $\li$ [resp.\ $\{b_k\}_{k=1}^\infty$] with $\hi\otimes\ki$ [resp.\ $\{h_n\otimes k_s\}$]. 
For example, when $\M\in\In(\Sigma,\,\ki,\,\hi)$ is concentrated on a point $y\in\Omega$, that is, the associated measure $\mu=\delta_y$ (the Dirac measure), then 
$$
\M(X,B)=\sum_{n,m}\sum_{s,t}B_{st}\<\psi_n^s|\psi_m^t\>\kb{h_n}{h_m},\qquad y\in X,
$$
where $\psi_n^s\in\hi\otimes\ki$. By choosing $\psi_n^s=\sqrt{\lambda_s}\,h_n\otimes k_s$, $\la_s\ge 0$, $\sum_s\la_s=1$ one gets a channel $\M\big(\Omega,B\big)=\big(\sum_s \lambda_s B_{ss}\big)I_{\hi}=\tr{\la B}I_{\hi}$ where $\la=\sum_s\la_s\kb{k_s}{k_s}$.
\end{remark}

\begin{remark}\rm
Let $\M(X,B)=Y^*(B\otimes\CHII X)Y=\tilde Y^*(B\otimes\CHII X)\tilde Y$ be two minimal Stinespring dilations for an instrument $\M$. Here $Y:\,\hi\to\ki\otimes\hd$ and $\tilde Y:\,\hi\to\ki\otimes\widetilde\hd$ are isometries where
$\hd=\int_\Omega^\oplus\hi_{n(x)}\d\mu(x)$ and $\widetilde\hd=\int_\Omega^\oplus\hi_{\tilde n(x)}\d\mu(x)$ are direct integral Hilbert spaces. Fix ON bases $\e$ and $\k$, and write
$Y=\sum_{m=1}^{\dim\hi}\sum_{t=1}^{\dim\ki}\kb{k_t\otimes\psi^t_m}{h_m}$ and 
$\tilde Y=\sum_{m=1}^{\dim\hi}\sum_{t=1}^{\dim\ki}\kb{k_t\otimes\tilde\psi^t_m}{h_m}$.
From Theorem \ref{kuusi} (see Example \ref{instruex}) it follows that there is a unitary mapping $U:\,\ki\otimes\hd\to\ki\otimes\widetilde\hd$ such that, on $\ki\otimes\widetilde\hd$,
$$
B\otimes\CHII X=U(B\otimes\CHII X)U^*, \hspace{5mm}B\in\lk,\;X\in\Sigma,
$$
and $\tilde Y=UY$. Hence, $U=I_\ki\otimes V$ where $V:\,\hd\to\widetilde\hd$ is a unitary map and $\hd\cong\widetilde\hd$. Moreover, $V$ commutes with the canonical spectral measure(s) and is thus decomposable:
$V=\int_\Omega^\oplus V(x)\d\mu(x)$ where $V(x):\,\hi_{n(x)}\to\hi_{\tilde n(x)}$ is unitary for almost all $x\in\Omega$. Especially, $\tilde n(x)=n(x)$ and $$\tilde\psi^t_m(x)=V(x)\psi^t_m(x)$$ for almost all $x\in\Omega$. If $\hd$ and $\widetilde\hd$ are embedded in $L^2(\mu,\li)$ as before, that is, $\hi_{\tilde n(x)}=\hi_{n(x)}={\lin_\C\{b_n\,|\,1\le n\le n(x)\}}$ one can view $V(x)$ as a change of the ON basis 
$\{b_n\}_{n=1}^{n(x)}$ of the fiber $\hi_{n(x)}$ at $x\in\Omega$: Let $d_k^s(x)$, $\A_k(x)$, and $\A^s(x)$ [resp.\ $\tilde d_k^s(x)$, $\tilde\A_k(x)$, and $\tilde\A^s(x)$] be the objects of Theorem \ref{th2} associated to the vectors $\psi^t_m$ [resp.\ $\tilde\psi^t_m$].
Then
\begin{align*}
d^s_k(x):=\sum_{n}\<\psi_n^s(x)|b_k\>h_n,\hspace{0.5cm}\tilde d^s_k(x)
&:=\sum_{n}\<\tilde \psi_n^s(x)|b_k\>h_n
=\sum_{n}\<\psi_n^s(x)|V(x)^*b_k\>h_n, \\
\A^s(x)=\sum_n\kb{\psi_n^s(x)}{h_n},\hspace{0.5cm}\tilde\A^s(x)&=\sum_n\kb{\tilde\psi_n^s(x)}{h_n}=V(x)\A^s(x),\\
\A_k(x)=\sum_s\kb{k_s}{b_k}\A^s(x),\hspace{0.5cm}
\tilde \A_l(x)&=\sum_s\kb{k_s}{b_l}\tilde \A^s(x)=\sum_s\kb{k_s}{b_l}V(x)\A^s(x) \\
&=\sum_{k}\<b_l|V(x)b_k\>\sum_s\kb{k_s}{b_k}\A^s(x)=\sum_{k}\<b_l|V(x)b_k\>\A_k(x)
\end{align*}
where
$\big(\<b_l|V(x)b_k\>\big)_{l,k=1}^{n(x)}$
 is a unitary matrix (compare to Remark 4 of \cite{Ch}).
Now, for example,
$$
\M(X,B)\supseteq\int_X\sum_{k=1}^{n(x)} \A_k(x)^*B\A_k(x)\d\mu(x)=\int_X\sum_{l=1}^{n(x)} \tilde\A_l(x)^*B\tilde\A_l(x)\d\mu(x).
$$
Note that if the bases $\e$ and $\k$ are changed to ON bases 
$\e'$ and $\k'$ of $\hi$ and $\ki$, respectively, then, e.g.\
$$
Y=\sum_{m=1}^{\dim\hi}\sum_{t=1}^{\dim\ki}\kb{k_t\otimes\psi^t_m}{h_m}
=\sum_{n=1}^{\dim\hi}\sum_{s=1}^{\dim\ki}\kb{k'_s\otimes\psi'^s_n}{h'_n}
$$
where 
$$
\psi'^s_n=\sum_{m=1}^{\dim\hi}\sum_{t=1}^{\dim\ki}
\<k'_s |k_t \>\psi^t_m\<h_m|h'_n\>.
$$
Now one can define the corresponding objects
$d'^s_k(x)$, $\A'^s(x)$, and $\A'_k(x)$ related to bases $\e'$ and $\k'$, and easily find their relations with $d^s_k(x)$, $\A^s(x)$, and $\A_k(x)$.
If the measure $\mu$ is replaced by an equivalent measure $\overline\mu$, $\d\overline\mu(x)=w(x)\d\mu(x)$, then the density (the Radon-Nikod\'ym derivative) $w$ can be absorbed into fibers $\hi_{n(x)}$ and the corresponding direct integral Hilbert spaces are unitarily equivalent.
Thus, we have seen that Theorem \ref{th2} is essentially independent of the choices of $\mu$, the ON bases, and the corresponding operators.
\end{remark}

\end{document}